\documentclass[sigconf]{acmart}

\AtBeginDocument{%
  \providecommand\BibTeX{{%
    \normalfont B\kern-0.5em{\scshape i\kern-0.25em b}\kern-0.8em\TeX}}}

\copyrightyear{2023}
\acmYear{2023}
\setcopyright{rightsretained}
\acmConference[KDD '23]{Proceedings of the 29th ACM SIGKDD Conference on Knowledge Discovery and Data Mining}{August 6--10, 2023}{Long Beach, CA, USA}
\acmBooktitle{Proceedings of the 29th ACM SIGKDD Conference on Knowledge Discovery and Data Mining (KDD '23), August 6--10, 2023, Long Beach, CA, USA}
\acmDOI{10.1145/3580305.3599504}
\acmISBN{979-8-4007-0103-0/23/08}

\usepackage{amsmath}
\usepackage{amsfonts}
\usepackage{amsthm}
\usepackage[ruled,vlined,shortend, linesnumbered]{algorithm2e} 
\usepackage{algpseudocode}
\usepackage{subcaption}
\usepackage{soul}
\usepackage{graphicx}

\usepackage{bbding} 
\usepackage{xcolor}
\usepackage{color, colortbl}
\usepackage{enumitem}
\usepackage{amsfonts}
\usepackage{courier}
\usepackage{multirow}
\usepackage{booktabs}
\usepackage{placeins}
\newsavebox{\measurebox}
\usepackage{numprint}

\usepackage{tikz}
\usetikzlibrary{automata,positioning}

\npthousandsep{,} 

\newtheorem{definition}{Definition}
\newtheorem{lemma}{Lemma}
\newtheorem{theorem}{Theorem}
\newtheorem*{theorem*}{Theorem}
\newtheorem{proposition}{Proposition}
\newtheorem*{proposition*}{Proposition}

\newenvironment{conditions}[1][where:]
  {#1 \begin{tabular}[t]{>{}l<{} @{{}={}} l}}
  {\end{tabular}\\[\belowdisplayskip]}

\usepackage{tikz}
\usetikzlibrary{automata,positioning}

\newcommand{\methodedge}{\textsc{AnoEdge}}
\newcommand{\methodgraph}{\textsc{AnoGraph}}
\newcommand{\sedanspot}{\textsc{SedanSpot}}
\newcommand{\densestream}{\textsc{DenseStream}}
\newcommand{\densealert}{\textsc{DenseAlert}}
\newcommand{\spotlight}{\textsc{SpotLight}}
\newcommand{\anomrank}{\textsc{AnomRank}}

\SetCommentSty{mycommfont}

\SetKwInput{KwInput}{Input}                
\SetKwInput{KwOutput}{Output}              

\usepackage{mathtools}
\usepackage{bbm}
\DeclarePairedDelimiter\ceil{\lceil}{\rceil}
\newcommand{\EE}{\mathbb{E}} 
\newcommand{\one}{\mathbbm{1}}  
\usepackage{hyperref}

\makeatletter
\gdef\@copyrightpermission{
  \begin{minipage}{0.3\columnwidth}
   {\includegraphics[width=0.90\textwidth]{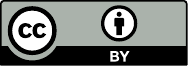}}
  \end{minipage}\hfill
  \begin{minipage}{0.7\columnwidth}
   \href{https://creativecommons.org/licenses/by/4.0/}{This work is licensed under a Creative Commons Attribution International 4.0 License.}
  \end{minipage}
  \vspace{5pt}
}
\makeatother

\begin{document}

\title{Sketch-Based Anomaly Detection in Streaming Graphs}


\author{Siddharth Bhatia}

\affiliation{%
  \institution{TurboML}
}
\email{siddharth@turboml.com}

\author{Mohit Wadhwa}
\affiliation{%
  \institution{}
}
\email{mailmohitwadhwa@gmail.com}

\author{Kenji Kawaguchi}
\affiliation{%
  \institution{National University of Singapore}
}
\email{kenji@comp.nus.edu.sg}

\author{Neil Shah}
\affiliation{%
  \institution{Snap Inc.}
}
\email{nshah@snap.com}

\author{Philip S. Yu}
\affiliation{%
  \institution{University of Illinois at Chicago}
}
\email{psyu@uic.edu}

\author{Bryan Hooi}
\affiliation{%
  \institution{National University of Singapore}
}
\email{bhooi@comp.nus.edu.sg}

\renewcommand{\shortauthors}{Siddharth Bhatia et al.}


\begin{abstract}
Given a stream of graph edges from a dynamic graph, how can we assign anomaly scores to edges and subgraphs in an online manner, for the purpose of detecting unusual behavior, using constant time and memory? For example, in intrusion detection, existing work seeks to detect either anomalous edges or anomalous subgraphs, but not both. In this paper, we first extend the count-min sketch data structure to a higher-order sketch. This higher-order sketch has the useful property of preserving the dense subgraph structure (dense subgraphs in the input turn into dense submatrices in the data structure). We then propose 4 online algorithms that utilize this enhanced data structure, which (a) detect both edge and graph anomalies; (b) process each edge and graph in constant memory and constant update time per newly arriving edge, and; (c) outperform state-of-the-art baselines on 4 real-world datasets. Our method is the first streaming approach that incorporates dense subgraph search to detect graph anomalies in constant memory and time.
\end{abstract}

\ccsdesc[500]{Computing methodologies~Anomaly detection}

\keywords{Anomaly Detection, Dynamic Graphs, Stream, Sketch}

\maketitle

\section{Introduction}

Consider an intrusion detection system, in which anomalous behavior can be described as an individual or a group of attackers making a large number of connections to some set of targeted machines to restrict accessibility or look for potential vulnerabilities. We can model this as a dynamic graph, where nodes correspond to machines, and each edge represents a timestamped connection from one machine to another. 

In this graph, edge anomalies include individual connections that are significantly more malicious than the rest of the connections in the network. In addition, anomalous behavior often also takes the form of a dense subgraph which could represent a group of malicious nodes that are communicating with each other in a way that is unusual compared to the rest of the graph, as shown in several real-world datasets in \citep{shin2017densealert,eswaran2018spotlight,bhatia2020midas}. Detecting both these edge and subgraph anomalies together provides valuable insights into the structure and behavior of the network and can help identify trends or patterns that might not be evident by considering only one type of anomaly.

Similarly, in a financial system, edge anomalies might include transactions that are significantly larger or more frequent than the rest of the transactions in the system. Subgraph anomalies might include groups of individuals or businesses who are significantly more likely to be involved in fraudulent activity than the rest of the system. Identifying both types of anomalies simultaneously can help detect fraudulent activity and protect against financial loss.

Thus, we ask the question: Given a stream of graph edges from a dynamic graph, how can we assign anomaly scores to both edges and subgraphs in an online manner, for the purpose of detecting unusual behavior, using constant memory and constant update time per newly arriving edge?

Several approaches \citep{akoglu2010oddball,chakrabarti2004autopart,hooi2017graph,jiang2016catching,kleinberg1999authoritative,shin2018patterns,tong2011non} aim to detect anomalies in graph settings. However, these approaches focus on static graphs, whereas many real-world graphs are time-evolving in nature. In streaming or online graph scenarios, some methods can detect the presence of anomalous edges, \citep{eswaran2018sedanspot,bhatia2020midas,belth2020mining,chang2021f}, while others can detect anomalous subgraphs \citep{shin2017densealert,eswaran2018spotlight,yoon2019fast}. However, all existing methods are limited to either anomalous edge or graph detection but not able to detect both kinds of anomalies, as summarized in Table \ref{tab:comparison}. As we discuss in Section \ref{sec:exp}, our approach outperforms existing methods in both accuracy and running time; and on both anomalous edge and subgraph detection scenarios. Moreover, our approach is the only streaming method that makes use of dense subgraph search to detect graph anomalies while only requiring constant memory and time.

We first extend the two-dimensional sketch to a \textit{higher-order sketch} to enable it to embed the relation between the source and destination nodes in a graph. A higher-order sketch has the useful property of preserving the dense subgraph structure; dense subgraphs in the input turn into dense submatrices in this data structure. Thus, the problem of detecting a dense subgraph from a large graph reduces to finding a dense submatrix in a constant size matrix, which can be achieved in constant time. The higher-order sketch allows us to propose several algorithms to detect both anomalous edges and subgraphs in a streaming manner. We introduce two edge anomaly detection methods, \methodedge-G, and \methodedge-L, and two graph anomaly detection methods \methodgraph, and \methodgraph-K, that use the same data structure to detect the presence of a dense submatrix, and consequently anomalous edges, or subgraphs respectively. All our approaches process edges and graphs in constant time, and are independent of the graph size, i.e., they require constant memory. We also provide theoretical guarantees on the higher-order sketch estimate and the submatrix density measure. In summary, the main contributions of our paper are:
\begin{enumerate}
    \item \textbf{Higher-Order Sketch (Section \ref{sec:sketch})}: We transform the dense subgraph detection problem into finding a dense submatrix (which can be achieved in constant time) by extending the count-min sketch (CMS) data structure to a higher-order sketch.
    \item \textbf{Streaming Anomaly Detection (Sections \ref{sec:edge},\ref{sec:graph})}: We propose four novel online approaches to detect anomalous edges and graphs in real-time, with constant memory and  update time. Moreover, this is the first streaming work that incorporates dense subgraph search to detect graph anomalies in constant memory/time.
    \item \textbf{Effectiveness (Section \ref{sec:exp})}: We outperform all state-of-the-art streaming edge and graph anomaly detection methods on four real-world datasets.
\end{enumerate}

{\bf Reproducibility}: Our code and datasets are available on \textcolor{blue}{\href{https://github.com/Stream-AD/AnoGraph}{https://github.com/Stream-AD/AnoGraph}}.


\begin{table*}[!tb]
\centering
\caption{Comparison of relevant anomaly detection approaches.}
\label{tab:comparison}
\resizebox{\linewidth}{!}{
\addtolength{\tabcolsep}{-2pt}
\begin{tabular}{@{}r|ccccc|ccc|c@{}}
\toprule
{Property} 
& {DenseStream}
& {SedanSpot}
& {MIDAS-R}
& {PENminer}
& {F-FADE}
& {DenseAlert}
& {SpotLight}
& {AnomRank}
& {\bf {Our Method}} \\
& (KDD'17) & (ICDM'20) & (AAAI'20) & (KDD'20) & (WSDM'21) & (KDD'17) & (KDD'18) & (KDD'19) \\\midrule
\textbf{Edge Anomaly} & \Checkmark & \Checkmark & \Checkmark & \Checkmark & \Checkmark & -- & -- & -- & \CheckmarkBold \\
\textbf{Graph Anomaly} & -- & -- & -- & -- & -- & \Checkmark & \Checkmark & \Checkmark & \CheckmarkBold \\
\textbf{Constant Memory} & -- & \Checkmark & \Checkmark & -- & \Checkmark & -- & \Checkmark & -- & \CheckmarkBold \\
\textbf{Constant Update Time} & -- & \Checkmark & \Checkmark & \Checkmark & \Checkmark & -- &  \Checkmark & -- & \CheckmarkBold \\
\textbf{Dense Subgraph Search} & \Checkmark & -- & -- & -- & -- & \Checkmark & -- & -- & \CheckmarkBold \\
\bottomrule
\end{tabular}}
\end{table*}

\section{Related Work}
Our work is closely related to areas like anomaly detection on graphs \citep{DBLP:conf/pakdd/ZhangLYFC19, DBLP:conf/sdm/BogdanovFMPRS13, 7836684, 10.1145/3139241, DBLP:journals/wias/BonchiBGS19, 7817049,Bojchevski2018BayesianRA,Yu2018NetWalkAF,Kumagai2021SemisupervisedAD,Liu2021AnomalyDI,Shao2018AnEF} and streams \citep{Bhatia2021MSTREAM,bhatia2022memstream,DBLP:conf/kdd/ManzoorLA18,tan2011fast,Jankov2017RealtimeHP,Zou2017NonparametricDO,Moshtaghi2015EvolvingFR, Siffer2017AnomalyDI,Togbe2020AnomalyDF,Zhang2020AugSplicingSB}, and streaming algorithms \citep{6544842,Wang2008ProcessingOM,Menon2007AnID, Zhao2011gSketchOQ,Shi2020HigherOrderCS}. Higher-order sketches are discussed in \citep{Shi2020HigherOrderCS}, however, they are restricted to count-sketches and non-graph settings.

\citep{chalapathy2019deep,pang2020deep,bhatia2021exgan} discuss deep learning based anomaly detection and \citep{Ma2020EfficientAF,Epasto2015EfficientDS,Sawlani2020NearoptimalFD,Mcgregor2015DensestSI,Esfandiari2018MetricSA} discus dense subgraph discovery, however, such approaches are unable to detect anomalies in a streaming manner. \citep{akoglu2010oddball,jiang2016catching,kleinberg1999authoritative,chakrabarti2004autopart,tong2011non,hooi2017graph,shin2018patterns} are limited to anomaly detection in static graphs. In this section, however, we limit our review only to methods detecting edge and graph anomalies in dynamic graphs; see \citep{akoglu2015graph} for an extensive survey.

{\bf Edge Stream Methods:} \textsc{HotSpot} \citep{yu2013anomalous} detects nodes whose egonets suddenly change.  RHSS \citep{ranshous2016scalable} focuses on sparsely-connected graph parts. CAD \citep{Sricharan} localizes anomalous changes using commute time distance measurement. More recently, \densestream\ \citep{shin2017densealert} maintains and updates a dense subtensor in a tensor stream. \sedanspot\ \citep{eswaran2018sedanspot} identifies edge anomalies based on edge occurrence, preferential attachment, and mutual neighbors. PENminer \citep{belth2020mining} explores the persistence of activity snippets, i.e., the length and regularity of edge-update sequences' reoccurrences. F-FADE \citep{chang2021f} aims to detect anomalous interaction patterns by factorizing their frequency. MIDAS \citep{bhatia2020midas,bhatia2022midas} identifies microcluster-based anomalies. However, all these methods are unable to detect graph anomalies.
	
{\bf Graph Stream Methods:} DTA/STA \citep{sun2006beyond} approximates the adjacency matrix of the current snapshot using matrix factorization. \textsc{Copycatch} \citep{beutel2013copycatch} spots near-bipartite cores where each node is connected to others in the same core densely within a short time. SPOT/DSPOT \citep{Siffer2017AnomalyDI} use extreme value theory to automatically set thresholds for anomalies. IncGM+ \citep{abdelhamid2017incremental} utilizes incremental method to process graph updates. More recently, \densealert\ identifies subtensors created within a short time and utilizes incremental method to process graph updates or subgraphs more efficiently. \spotlight\ \citep{eswaran2018spotlight} discovers anomalous graphs with dense bi-cliques, but uses a randomized approach without any search for dense subgraphs. \anomrank\ \citep{yoon2019fast}, inspired by PageRank \citep{page1999the}, iteratively updates two score vectors and computes anomaly scores. However, these methods are slow and do not detect edge anomalies. Moreover, they do not search for dense subgraphs in constant memory/time.

\section{Problem}
Let $\mathscr{E} = \{e_1, e_2, \cdots\}$ be a stream of weighted edges from a time-evolving graph $\mathcal{G}$. Each arriving edge is a tuple $e_i = (u_i, v_i, w_i, t_i)$ consisting of a source node $u_i \in \mathcal{V}$, a destination node $v_i \in \mathcal{V}$, a weight $w_i$, and a time of occurrence $t_i$, the time at which the edge is added to the graph. For example, in a network traffic stream, an edge $e_i$ could represent a connection made from a source IP address $u_i$ to a destination IP address $v_i$ at time $t_i$. We do not assume that the set of vertices $\mathcal{V}$ is known a priori: for example, new IP addresses or user IDs may be created over the course of the stream.

We model $\mathcal{G}$ as a directed graph. Undirected graphs can be handled by treating an incoming undirected edge as two simultaneous directed edges, one in each direction. We also allow $\mathcal{G}$ to be a multigraph: edges can be created multiple times between the same pair of nodes. Edges are allowed to arrive simultaneously: i.e. $t_{i+1} \ge t_i$, since in many applications $t_i$ is given as a discrete time tick.

The desired properties of our algorithm are as follows:

\begin{itemize}
\item {\bf Detecting Anomalous Edges:} To detect whether the edge is part of an anomalous subgraph in an online manner. Being able to detect anomalies at the finer granularity of edges allows early detection so that recovery can be started as soon as possible and the effect of malicious activities is minimized.

\item {\bf Detecting Anomalous Graphs:} To detect the presence of an unusual subgraph (consisting of edges received over a period of time) in an online manner, since such subgraphs often correspond to unexpected behavior, such as coordinated attacks. 

\item {\bf Constant Memory and Update Time:} To ensure scalability, memory usage and update time should not grow with the number of nodes or the length of the stream. Thus, for a newly arriving edge, our algorithm should run in constant memory and update time.

\end{itemize}

\section{Higher-Order Sketch \& Notations}
\label{sec:sketch}


Count-min sketches (CMS) \citep{cormode2005improved} are popular streaming data structures used by several online algorithms \citep{Mcgregor2014GraphSA}. CMS uses multiple hash functions to map events to frequencies, but unlike a hash table uses only sub-linear space, at the expense of overcounting some events due to collisions. Frequency is approximated as the minimum over all hash functions. CMS, shown in Figure \ref{fig:sketch}(a), is represented as a two-dimensional matrix where each row corresponds to a hash function and hashes to the same number of buckets (columns).

\begin{figure}[!b]
        \centering
        \includegraphics[width=0.85\columnwidth]{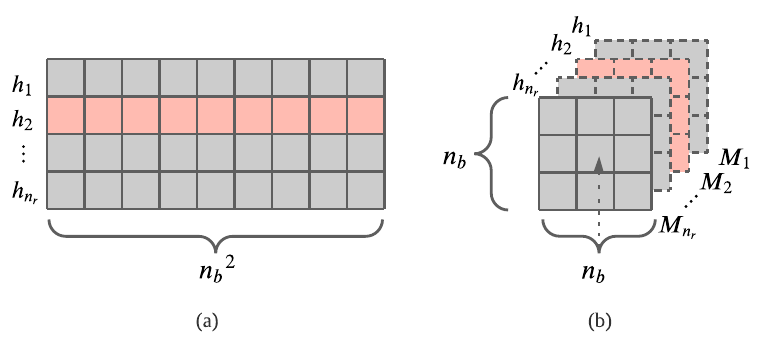} 
        \caption{(a) Original CMS with $n_b^{2}$ buckets for each hash function (b) Higher-order CMS with $n_b$ x $n_b$ buckets for each hash function.}
        \label{fig:sketch}
\end{figure}

\begin{figure}[!tb]
  \centering
  \includegraphics[width=0.76\columnwidth]{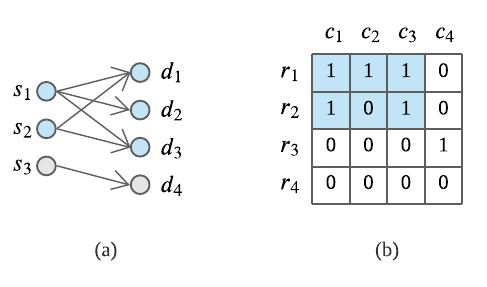}
  \caption{(a) Dense subgraph in the original graph between source nodes $s_1, s_2$, and destination nodes $d_1, d_2, d_3$ is transformed to a Dense submatrix between rows $r_1,r_2$, and columns $c_1, c_2, c_3$ in the H-CMS (b).}
  \label{fig:hcms}
\end{figure}

\begin{algorithm}[!hb]
    \DontPrintSemicolon
    \SetNoFillComment
    \SetKwFunction{procinit}{INITIALIZE H-CMS}
    \SetKwFunction{procreset}{RESET H-CMS}
    \SetKwFunction{procupdate}{UPDATE H-CMS}
    \SetKwFunction{procdecay}{DECAY H-CMS}
    
	\setcounter{AlgoLine}{0}
    \SetKwProg{myproc}{Procedure}{}{}
    \myproc{\procinit{$n_r$, $n_b$}}{
        \For{$r\gets1$ ... $n_r$}{
            $h_r: \mathcal{V} \rightarrow [0, n_b)$ \tcp*{hash vertex}
            $M_r \leftarrow [0]_{n_b \times n_b}$
        }
    }
    \setcounter{AlgoLine}{0}
    \SetKwProg{myproc}{Procedure}{}{}
    \myproc{\procreset{$n_r$, $n_b$}}{
        \For{$r\gets1$ ... $n_r$}{
            $\mathcal{M}_r \leftarrow [0]_{n_b \times n_b}$ \tcp*{reset to zero matrix} 
        }
    }
    \setcounter{AlgoLine}{0}
    \SetKwProg{myproc}{Procedure}{}{}
    \myproc{\procupdate{$u, v, w$}}{
        \For{$r\gets1$ ... $n_r$} {
            $\mathcal{M}_r[h_r(u)][h_r(v)] \leftarrow \mathcal{M}_r[h_r(u)][h_r(v)] + w$
        }
    }
    \setcounter{AlgoLine}{0}
    \SetKwProg{myproc}{Procedure}{}{}
    \myproc{\procdecay{$\alpha$}}{
        \For{$r\gets1$ ... $n_r$}{
            $\mathcal{M}_r \leftarrow \alpha * \mathcal{M}_r$ \tcp*{decay factor: $\alpha$ }
        }
    }
    
	\caption{H-CMS Operations}
	\label{alg:hcmsoperations}
\end{algorithm}

We introduce a Higher-order CMS (H-CMS) data structure where each hash function maps multi-dimensional input to a generic tensor instead of mapping it to a row vector. H-CMS enhances CMS by separately hashing the individual components of an entity thereby maintaining more information. Figure \ref{fig:sketch}(b) shows a 3-dimensional H-CMS that can be used to hash two-dimensional entities such as graph edges to a matrix. The source node is hashed to the first dimension and the destination node to the other dimension of the sketch matrix, as opposed to the original CMS that will hash the entire edge to a one-dimensional row vector (Figure \ref{fig:sketch}(a)).

We use a 3-dimensional H-CMS (operations described in Algorithm \ref{alg:hcmsoperations}) where the number of hash functions is denoted by $n_r$, and matrix $\mathcal{M}_j$ corresponding to $j$-th hash function $h_j$ is of dimension $n_b \times n_b$, i.e., a square matrix. For each $j \in [n_r]$, the $j$-th hash function denoted by $h_{j}(u,v)$ maps an edge $(u, v)$ to a matrix index $(h'_{j}(u), h''_{j}(v))$, i.e., the source node is mapped to a row index and the destination node is mapped to a column index. That is, $h_{j}(u,v)=(h'_{j}(u), h''_{j}(v))$. Therefore, each matrix in a 3-dimensional H-CMS captures the essence of a graph adjacency matrix. Dense subgraph detection can thus be transformed into a dense submatrix detection problem (as shown in Figure \ref{fig:hcms}) where the size of the matrix is a small constant, independent of the number of edges or the graph size.

For any $(u,v)$, let $y(u,v)$ be the true count of $(u,v)$ observed thus far and $\hat y(u,v)=\min_{j \in [n_r]} \mathcal{M}_j[h_j'(u)][h_j''(v)]$ be the estimate of the count via the 3-dimensional H-CMS. Since the H-CMS can overestimate the count by possible collisions (but not underestimate because we update and keep all the counts for every hash function), we have $y(u,v)\le \hat y(u,v)$. We define $M$ to be  the number of all observations so far; i.e., $M=\sum_{u,v} y(u,v)$.
The following theorem shows that the 3-dimensional H-CMS has the  estimate guarantees similarly to  the CMS:
\begin{theorem} \label{thm:1}
(Proof in Appendix \ref{app:proofshcms})
For all $k \in [n_r]$, let $h_k(u,v) = (h'_k(u), h''_k(v))$ where each of hash functions $h'_k$ and $h''_k$ is chosen uniformly at random from a pairwise-independent family. Here, we allow both cases of $h'=h''$ and $h'\neq h''$. Fix $\delta>0$ and set  $n_r = \ceil*{\ln \frac{1}{\delta}}$ and  $n_b^{}=\ceil{\frac{e}{\epsilon}}$.  Then, with probability at least $1-\delta$, $ \hat y(u,v)\le y(u,v)+\epsilon M$.
\end{theorem}

Theorem \ref{thm:1} shows that we have the estimate guarantee even if we use the same hash function for both  the source nodes and the destination node (i.e., $h'=h''$). Thus, with abuse of notation, we write $h(u,v)=(h(u), h(v))$ when $h'=h''$ by setting $h=h'=h''$ on the right-hand side. On the other hand, in the case of $h' \neq h''$, it would be possible to improve the estimate guarantee in Theorem \ref{thm:1}. For example, if we can make $h$ to be chosen uniformly at random from a  weakly universal set of hash functions (by defining corresponding families of distributions for  $h'$ and  $h''$ under some conditions), then we can set $n_b=\ceil{\sqrt{\frac{e}{\epsilon}}}$ to have the same estimate guarantee as that of  Theorem \ref{thm:1} based on the proof of Theorem \ref{thm:1}. The analysis for such a potential improvement is left for future work as an open problem.

Frequently used symbols are discussed in Table \ref{tab:symbols}, and we leverage the subgraph density measure discussed in \citep{khuller2009finding} to define the submatrix $(S_x, T_x)$ density.

\begin{definition}
Given matrix $\mathcal{M}$, density of a submatrix of $\mathcal{M}$ represented by $S_x \subseteq S$ and $T_{x} \subseteq T$, is: 
\begin{equation}
\mathcal{D}(\mathcal{M}, S_x, T_x) = \frac{\sum_{s \in S_x}\sum_{t \in T_x}\mathcal{M}[s][t]}{\sqrt{|S_x||T_x|}}
\end{equation}
\label{def:density}
\end{definition}	



\begin{table}[!tb]
\centering
\caption{Table of symbols.}
\label{tab:symbols}
\begin{tabular}{>{\centering\arraybackslash}p{2.2cm}|p{6cm}}
\toprule
\textbf{Symbol} & \textbf{Definition} \\
\midrule
$n_r$ & number of hash functions \\
$n_b$ & number of buckets \\
$h(u)$ & hash function $u \rightarrow [0, n_b)$\\ \midrule
$\mathcal{M}$ & a square matrix of dimensions $n_b \times n_b$\\
$\mathcal{M}[i][j]$ & element at row index i and column index j\\ \midrule
$S$ & set of all row indices\\
$S_{cur}$ & set of current submatrix row indices \\
$S_{rem}$ & set of remaining row indices \\
$T$ & set of all column indices \\
$T_{cur}$ & set of current submatrix column indices \\
$T_{rem}$ & set of remaining column indices \\
$[z]$ & set of all integers in the range $[1, z]$ \\\midrule
$\mathcal{D}(\mathcal{M}, S_x, T_x)$ & density of submatrix ($S_x$, $T_x$) \\ 
$\mathcal{E}(\mathcal{M}, S_x, T_x)$ & sum of elements of submatrix ($S_x$, $T_x$)\\ 
$\mathcal{R}(\mathcal{M}, u, T_x)$ & submatrix row-sum \\
& i.e. sum of elements of submatrix ($\{u\}$, $T_x$) \\ 
$\mathcal{C}(\mathcal{M}, S_x, v)$ & submatrix column-sum \\
& i.e. sum of elements of submatrix ($S_x$, $\{v\}$)\\ 
$\mathcal{L}(\mathcal{M}, u, v, S_x, T_x)$ & likelihood of index $(u, v)$ w.r.t. submatrix $(S_x, T_x)$ \\ 
$d_{max}$ & maximum reported submatrix density \\
\bottomrule
\end{tabular}
\end{table}



\section{Edge Anomalies}
\label{sec:edge}

In this section, using the H-CMS data structure, we propose \methodedge-G and \methodedge-L to detect edge anomalies by checking whether the received edge when mapped to a sketch matrix element is part of a dense submatrix. \methodedge-G finds a \textbf{G}lobal dense submatrix and \methodedge-L maintains and updates a \textbf{L}ocal dense submatrix around the matrix element.

\subsection{\methodedge-G}

\methodedge-G, as described in Algorithm \ref{alg:AnoEdge-G}, maintains a \emph{temporally decaying} H-CMS, i.e. whenever 1 unit of time passes, we multiply all the H-CMS counts by a fixed factor $\alpha$ (lines 2,4). This decay simulates the gradual `forgetting' of older, and hence, more outdated information. When an edge $(u, v)$ arrives, $u$, $v$ are mapped to matrix indices $h(u)$, $h(v)$ respectively for each hash function $h$, and the corresponding H-CMS counts are updated (line 5). \textsc{Edge-Submatrix-Density} procedure (described below) is then called to compute the density of a dense submatrix around $(h(u), h(v))$. Density is reported as the anomaly score for the edge; a larger density implies that the edge is more likely to be anomalous.

\textsc{Edge-Submatrix-Density} procedure calculates the density of a dense submatrix around a given index $(h(u), h(v))$. A $1 \times 1$ submatrix represented by $S_{cur}$ and $T_{cur}$, is initialized with row-index $h(u)$ and column index $h(v)$ (line 9). The submatrix is iteratively expanded by greedily selecting a row $u_p$ from $S_{rem}$ (or a column $v_p$ from $T_{rem}$) that obtains the maximum row (or column) sum with the current submatrix (lines 11,12). This selected row $u_p$ (or column $v_p$) is removed from $S_{rem}$ (or $T_{rem}$), and added to $S_{cur}$ (or $T_{cur}$) (lines 14,16). The process is repeated until both $S_{rem}$ and $T_{rem}$ are empty (line 10). Density of the current submatrix is computed at each iteration of the submatrix expansion process and the maximum over all greedily formed submatrix densities is returned (lines 17,18).

\begin{algorithm}[!htb]
\caption{\methodedge-G Scoring}
\label{alg:AnoEdge-G}
    \DontPrintSemicolon
    \SetNoFillComment
    \SetKwFunction{algo}{algo}\SetKwFunction{procf}{\textsc{Edge-Submatrix-Density}}\SetKwFunction{procs}{\textsc{AnoEdge-G}}

    
    \KwInput{Stream $\mathscr{E}$ of edges over time}
    \KwOutput{Anomaly score per edge}
    
    \setcounter{AlgoLine}{0}
    \SetKwProg{myproc}{Procedure}{}{}
    \myproc{\procs{$\mathscr{E}$}}{
    Initialize H-CMS matrix $\mathcal{M}$ for edge count \\
    \While{new edge $e = (u, v, w, t) \in \mathscr{E}$ is received} {
        \tcc{decay count}
        Temporal decay H-CMS with timestamp change \\
        Update H-CMS matrix $\mathcal{M}$ for new edge $(u, v)$ with value $w$ \tcp*{update count}
        \textbf{output} $score(e) \leftarrow$ \textsc{Edge-Submatrix-Density}($\mathcal{M}, h(u), h(v)$) 
    }
    }
    
    \SetKwProg{myproc}{Procedure}{}{}
    \myproc{\procf{$\mathcal{M}$, $u$, $v$}}{
    $S \leftarrow [n_b]; \enspace T \leftarrow [n_b]; \enspace S_{cur} \leftarrow \{u\}; \enspace T_{cur} \leftarrow \{v\}; \enspace S_{rem} \leftarrow S/\{u\}; \enspace T_{rem} \leftarrow T/\{v\}$ \;
    $d_{max} \leftarrow \mathcal{D}(\mathcal{M}, S_{cur}, T_{cur})$ \;
    \While{$S_{rem} \neq \emptyset \enspace \vee \enspace T_{rem} \neq \emptyset$} {
        \tcc{submatrix max row-sum index}
        $u_p \leftarrow \operatorname*{argmax}_{s_p \in S_{rem}} \mathcal{R}(\mathcal{M}, s_p, T_{cur})$ \\
        \tcc{submatrix max column-sum index}
        $v_p \leftarrow \operatorname*{argmax}_{t_p \in T_{rem}} \mathcal{C}(\mathcal{M}, S_{cur}, t_p)$ \\
        \If{$\mathcal{R}(\mathcal{M}, u_p, T_{cur}) > \mathcal{C}(\mathcal{M}, S_{cur}, v_p)$} {
            $S_{cur} \leftarrow S_{cur}\cup\{u_p\}; \enspace S_{rem} \leftarrow S_{rem}/\{u_p\} $\;
        } 
        \Else {
             $T_{cur} \leftarrow T_{cur}\cup\{v_p\}; \enspace T_{rem} \leftarrow T_{rem}/\{v_p\} $\;
        }
        $d_{max} \leftarrow max(d_{max}, \mathcal{D}(\mathcal{M}, S_{cur}, T_{cur}))$ \;
    }
    \KwRet $d_{max}$ \tcp*{dense submatrix density}} 

\end{algorithm}

\begin{proposition}\label{thm:AnoEdge-G-time}
(Proof in Appendix \ref{app:proofsAnoEdge-G})
Time complexity of Algorithm \ref{alg:AnoEdge-G} is $O(|\mathscr{E}|*n_r*n_b^2)$ \footnote{This is for processing all edges; the time per edge is constant.}. Memory complexity of Algorithm \ref{alg:AnoEdge-G} is $O(n_r*n_b^2)$.
\end{proposition}

\subsection{\methodedge-L}
Inspired by Definition \ref{def:density}, we define the likelihood measure of a matrix index $(h(u), h(v))$ with respect to a submatrix $(S_x, T_x)$, as the sum of the elements of submatrix $(S_x, T_x)$ that either share row with index $h''(v)$ or column with index $h'(u)$ divided by the total number of such elements.
\begin{definition}
Given matrix $\mathcal{M}$, likelihood of an index $h(u, v)$ with respect to a submatrix represented by $S_x \subseteq S$ and $T_{x} \subseteq T$, is:
\begin{equation}
\mathcal{L}(\mathcal{M}, u, v, S_x, T_x) = \frac{\sum_{(s, t) \; \in \; \; S_x \times \{h(v)\} \; \cup \; \{h(u)\} \times {T_x}}\mathcal{M}[s][t]}{|S_x \times \{h(v)\} \; \cup \; \{h(u)\} \times {T_x}|}
\end{equation}
\label{def:likelihood}
\end{definition}

\methodedge-L, as described in Algorithm \ref{alg:AnoEdge-L}, maintains a temporally decaying H-CMS to store the edge counts. We also initialize a mutable submatrix of size $1 \times 1$ with a random element, and represent it as $(S_{cur}, T_{cur})$. As we process edges, we greedily update $(S_{cur}, T_{cur})$ to maintain it as a dense submatrix. When an edge arrives, H-CMS counts are first updated, and the received edge is then used to check whether to \emph{expand} the current submatrix (line 7). If the submatrix density increases upon the addition of the row (or column), then the row-index $h(u)$ (or column-index $h(v)$) is added to the current submatrix, $(S_{cur}, T_{cur})$. To remove the row(s) and column(s) decayed over time, the process iteratively selects the row (or column) with the minimum row-sum (or column-sum) until removing it increases the current submatrix density. This ensures that the current submatrix is as \emph{condensed} as possible (line 9). As defined in Definition \ref{def:likelihood}, \methodedge-L computes the likelihood score of the edge with respect to $(S_{cur}, T_{cur})$ (line 10). A higher likelihood measure implies that the edge is more likely to be anomalous.

\begin{algorithm}[!htb]
\caption{\methodedge-L Scoring}
    \label{alg:AnoEdge-L}
    \DontPrintSemicolon
    \SetNoFillComment
    \SetKwFunction{algo}{algo}\SetKwFunction{procs}{\textsc{AnoEdge-L}}
    
    \KwInput{Stream $\mathscr{E}$ of edges over time}
    \KwOutput{Anomaly score per edge}
    \setcounter{AlgoLine}{0}
    \SetKwProg{myproc}{Procedure}{}{}
    \myproc{\procs{$\mathscr{E}$}}{
    Initialize H-CMS matrix $\mathcal{M}$ for edges count \\
    \tcc{mutable submatrix}
    Initialize a randomly picked $1 \times 1$ submatrix $(S_{cur}, T_{cur})$ \\
    \While{new edge $e = (u, v, w, t) \in \mathscr{E}$ is received} {
        \tcc{decay count}
        Temporal decay H-CMS with timestamp change \\
        Update H-CMS matrix $\mathcal{M}$ for new edge $(u, v)$ with value $w$ \tcp*{update count}
        \textbf{$\triangleright$ Check and Update Submatrix:} \;
        Expand $(S_{cur}, T_{cur})$ \tcp*{expand submatrix}
        Condense $(S_{cur}, T_{cur})$ \tcp*{condense submatrix}
        \tcc{likelihood score from Definition \ref{def:likelihood}}
        \textbf{output} $score(e) \leftarrow \mathcal{L}(\mathcal{M}, h(u), h(v), S_{cur}, T_{cur})$ 
    }
    }
    
\end{algorithm}

\begin{proposition}\label{thm:AnoEdge-L-time}
(Proof in Appendix \ref{app:proofsAnoEdge-L}) Time complexity of Algorithm \ref{alg:AnoEdge-L} is $O(n_r*n_b^2 + |\mathscr{E}|*n_r*n_b)$. Memory complexity of Algorithm \ref{alg:AnoEdge-L} is $O(n_r*n_b^2)$.
\end{proposition}

\section{Graph Anomalies}
\label{sec:graph}
We now propose \methodgraph\ and \methodgraph-K to detect graph anomalies by first mapping the graph to a higher-order sketch, and then checking for a dense submatrix. These are the first streaming algorithms that make use of dense subgraph search to detect graph anomalies in constant memory and time. \methodgraph\ greedily finds a dense submatrix with a 2-approximation guarantee on the density measure. \methodgraph-K leverages \textsc{Edge-Submatrix-Density} from Algorithm \ref{alg:AnoEdge-G} to greedily find a dense submatrix around \textbf{$K$} strategically picked matrix elements performing equally well in practice.

\subsection{\methodgraph}
\methodgraph, as described in Algorithm \ref{alg:AnoGraph}, maintains an H-CMS to store the edge counts that are reset whenever a new graph arrives. The edges are first processed to update the H-CMS counts. \textsc{\methodgraph-Density} procedure (described below) is then called to find the dense submatrix. \methodgraph\ reports anomaly score as the density of the detected (dense) submatrix; a larger density implies that the graph is more likely to be anomalous.

\textsc{\methodgraph-Density} procedure computes the density of a dense submatrix of matrix $\mathcal{M}$. The current dense submatrix is initialised as matrix $\mathcal{M}$ and then the row (or column) from the current submatrix with minimum row (or column) sum is greedily removed. This process is repeated until $S_{cur}$ and $T_{cur}$ are empty (line 11). The density of the current submatrix is computed at each iteration of the submatrix expansion process and the maximum over all densities is returned (lines 18, 19). 

Algorithm \ref{alg:AnoGraph} is a special case of finding the densest subgraph in a directed graph problem \citep{khuller2009finding} where the directed graph is represented as an adjacency matrix and detecting the densest subgraph essentially means detecting dense submatrix. We now provide a guarantee on the density measure.


\begin{lemma}\label{lemma:2approx} Let $S^*$ and $T^*$ be the optimum densest sub-matrix solution of $\mathcal{M}$ with density $\mathcal{D}(\mathcal{M}, S^*, T^*) = d_{opt}$. Then $\forall u \in S^*$ and $\forall v \in T^*$,
\begin{equation}
    \mathcal{R}(\mathcal{M}, u, T^*) \ge \tau_{S^*}; \quad \mathcal{C}(\mathcal{M}, S^*, v) \ge \tau_{T^*}
\end{equation}
\begin{conditions}
    $\tau_{S^*}$ & $\mathcal{E}(\mathcal{M}, S^*, T^*)\left(1- \sqrt{1 - \frac{1}{|S^*|}}\right)$, \\
    $\tau_{T^*}$ & $\mathcal{E}(\mathcal{M}, S^*, T^*)\left(1- \sqrt{1 - \frac{1}{|T^*|}}\right)$ \\
\end{conditions}
\end{lemma}
\begin{proof}
Leveraging the proof from \citep{khuller2009finding}, let's assume that $\exists u \in S^*$ with $\mathcal{R}(\mathcal{M}, u, T^*) < \tau_{S^*}$. Density of submatrix after removing $u = \frac{\mathcal{E}(\mathcal{M}, S^*, T^*) - \mathcal{R}(\mathcal{M}, u, T^*)}{\sqrt{(|S^*-1|)|T^*|}}$ which is greater than $\frac{\mathcal{E}(\mathcal{M}, S^*, T^*) - \tau_{S^*}}{\sqrt{(|S^*-1|)|T^*|}}=d_{opt}$, and that is not possible. Hence, $\mathcal{R}(\mathcal{M}, u, T^*) \ge \tau_{S^*}$. $\mathcal{C}(\mathcal{M}, S^*, v) \ge \tau_{T^*}$ can be proved in a similar manner.
\end{proof}

\begin{theorem}\label{thm:2approx-supp}
\textsc{\methodgraph-Density} procedure in Algorithm \ref{alg:AnoGraph} achieves a 2-approximation guarantee for the densest submatrix problem.
\end{theorem}
\begin{proof} Leveraging the proof from \citep{khuller2009finding}, we greedily remove the row (or column) with minimum row-sum (or column-sum). At some iteration of the greedy process, $\;\forall u \in S_{cur}; \forall v \in T_{cur}$, $\;\mathcal{R}(\mathcal{M}, u, T_{cur}) \ge \tau_{S^*}$ and $\mathcal{C}(\mathcal{M}, S_{cur}, v) \ge \tau_{T^*}$. Therefore, $\mathcal{E}(\mathcal{M}, S_{cur}, T_{cur}) \ge |S_{cur}|\tau_{S^*}$ and $\mathcal{E}(\mathcal{M}, S_{cur}, T_{cur}) \ge |T_{cur}|\tau_{T^*}$. This implies that the density $\mathcal{D}(\mathcal{M}, S_{cur}, T_{cur}) \ge \sqrt{\frac{|S_{cur}|\tau_{S^*}|T_{cur}|\tau_{T^*}}{|S_{cur}||T_{cur}|}} = \sqrt{\tau_{S^*}\tau_{T^*}}$. Putting values of $\tau_{S^*}$ and $\tau_{T^*}$ from Lemma \ref{lemma:2approx}, and setting $|S^*| = \frac{1}{\sin^2\alpha}$, $|T^*| = \frac{1}{\sin^2\beta}$, we get $\mathcal{D}(\mathcal{M}, S_{cur}, T_{cur}) \ge \frac{\mathcal{E}(\mathcal{M, S^*, T^*})}{\sqrt{|S^*||T^*|}}\frac{\sqrt{(1-\cos\alpha)(1-\cos\beta)}}{\sin\alpha\sin\beta} \ge \frac{d_{opt}}{2\cos\frac{\alpha}{2}\cos\frac{\beta}{2}} \ge \frac{d_{opt}}{2}$.
\end{proof}

\begin{algorithm}[!htb]
\caption{\methodgraph Scoring}
\label{alg:AnoGraph}
    \DontPrintSemicolon
    \SetNoFillComment
    \SetKwFunction{algo}{algo}\SetKwFunction{procf}{\textsc{\methodgraph-Density}}\SetKwFunction{procs}{\methodgraph}

    \KwInput{Stream $\mathscr{G}$ of edges over time}
    \KwOutput{Anomaly score per graph}
    
    \setcounter{AlgoLine}{0}
    \SetKwProg{myproc}{Procedure}{}{}
    \myproc{\procs{$\mathscr{G}$}}{
    Initialize H-CMS matrix $\mathcal{M}$ for graph edges count \\
    \While{new graph $G \in \mathscr{G}$ is received} {
        \tcc{reset count}
        Reset H-CMS matrix $\mathcal{M}$ for graph $G$ \\
        \For{edge $e = (u, v, w, t) \in G$}{
            Update H-CMS matrix $\mathcal{M}$ for edge $(u, v)$ with value $w$ \tcp*{update count}
        }
        \tcc{anomaly score}
        \textbf{output} $score(G) \leftarrow$ \methodgraph\textsc{-Density}($\mathcal{M}$) 
    }}
    
    \myproc{\procf{$\mathcal{M}$}}{
    $S_{cur} \leftarrow [n_b]; \enspace T_{cur} \leftarrow [n_b]$ \tcp*{initialize to size of $\mathcal{M}$}
    $d_{max} \leftarrow \mathcal{D}(\mathcal{M}, S_{cur}, T_{cur})$ \;
    \While{$S_{cur} \neq \emptyset \enspace \vee \enspace T_{cur} \neq \emptyset$} {
        \tcc{submatrix min row-sum index}
        $u_p \leftarrow \operatorname*{argmin}_{s_p \in S_{cur}} \mathcal{R}(\mathcal{M}, s_p, T_{cur})$ \\
        \tcc{submatrix min column-sum index}
        $v_p \leftarrow \operatorname*{argmin}_{t_p \in T_{cur}} \mathcal{C}(\mathcal{M}, S_{cur}, t_p)$ \\
        \If{$\mathcal{R}(\mathcal{M}, u_p, T_{cur}) < \mathcal{C}(\mathcal{M}, S_{cur}, v_p)$} {
            $S_{cur} \leftarrow S_{cur}/\{u_p\}$ \tcp*{remove row}
        } 
        \Else {
             $T_{cur} \leftarrow T_{cur}/\{v_p\}$ \tcp*{remove column}
        }
        $d_{max} \leftarrow max(d_{max}, \mathcal{D}(\mathcal{M}, S_{cur}, T_{cur}))$ \;
    }
    \KwRet $d_{max}$ \tcp*{dense submatrix density}}
\end{algorithm}



\begin{proposition}\label{thm:AnoGraph-time}
(Proof in Appendix \ref{app:proofsAnoGraph}) Time complexity of Algorithm \ref{alg:AnoGraph} is $O(|\mathscr{G}|*n_r*n_b^2 + |\mathscr{E}|*n_r)$. Memory complexity of Algorithm \ref{alg:AnoGraph} is $O(n_r*n_b^2)$.
\end{proposition}


\subsection{\methodgraph-K}
Similar to \methodgraph, \methodgraph-K maintains an H-CMS which is reset whenever a new graph arrives. It uses the \textsc{\methodgraph-K-Density} procedure (described below) to find the dense submatrix. \methodgraph-K is summarised in Algorithm \ref{alg:AnoGraph-K}.

\textsc{\methodgraph-K-Density} computes the density of a dense submatrix of matrix $\mathcal{M}$. The intuition comes from the heuristic that the matrix elements with a higher value are more likely to be part of a dense submatrix. Hence, the approach considers $K$ largest elements of the matrix $\mathcal{M}$ and calls \textsc{Edge-Submatrix-Density} from Algorithm \ref{alg:AnoEdge-G} to get the dense submatrix around each of those elements (line 13). The maximum density over the considered $K$ dense submatrices is returned.

\begin{algorithm}[!htb]
\caption{\methodgraph-K Scoring}
\label{alg:AnoGraph-K}
    \DontPrintSemicolon
    \SetNoFillComment
    \SetKwFunction{algo}{algo}\SetKwFunction{procf}{\methodgraph-K-Density}\SetKwFunction{procs}{\textsc{\methodgraph-K}}
    
    \KwInput{Stream $\mathscr{G}$ of edges over time}
    \KwOutput{Anomaly score per graph}
    
    \setcounter{AlgoLine}{0}
    \SetKwProg{myproc}{Procedure}{}{}
    \myproc{\procs{$\mathscr{G}, K$}}{
    Initialize H-CMS matrix $\mathcal{M}$ for graph edges count \\
    \While{new graph $G \in \mathscr{G}$ is received} {
        \tcc{reset count}
        Reset H-CMS matrix $\mathcal{M}$ for graph $G$ \\
        \For{edge $e = (u, v, w, t) \in G$}{
            Update H-CMS matrix $\mathcal{M}$ for edge $(u, v)$ with value $w$ \tcp*{update count}
        }
        \tcc{anomaly score}
        \textbf{output} $score(G) \leftarrow$ \methodgraph\textsc{-K-Density}($\mathcal{M}, K$) 
    }}
    
    \myproc{\procf{$\mathcal{M}$, $K$}}{
    $B \leftarrow [n_b] \times [n_b]$ \tcp*{set of all indices}
    $d_{max} \leftarrow 0$ \;
    \For{$j\gets1$ ... $K$}{
        \tcc{pick the max element}
        $u_p, v_p \leftarrow \operatorname*{argmax}_{(s_p, t_p) \in B} \mathcal{M}[s_p][t_p]$ \\
        $d_{max} \leftarrow max(d_{max}, \textsc{Edge-Submatrix-Density}({\mathcal{M}, u_p, v_p}))$ \;
        $B \leftarrow B/\{(u_p, v_p)\}$ \tcp*{remove max element index}
    }
    \KwRet $d_{max}$ \tcp*{dense submatrix density}}
\end{algorithm}

\begin{proposition}\label{thm:AnoGraph-K-time}
(Proof in Appendix \ref{app:proofsAnoGraph-K}) Time complexity of Algorithm \ref{alg:AnoGraph-K} is $O(|\mathscr{G}|*K*n_r*n_b^2 + |\mathscr{E}|*n_r)$. Memory complexity of Algorithm \ref{alg:AnoGraph-K} is $O(n_r*n_b^2)$.
\end{proposition}


\section{Experiments}
\label{sec:exp}


\begin{table}[!htb]
		\centering
		\caption{Statistics of the datasets.}
		\begin{tabular}{lrrrrr}
			\toprule
			\textbf{Dataset} & $|V|$           & $|E|$             & $|T|$  \\
			\midrule
			DARPA               & \numprint{25525}  & \numprint{4554344}  & \numprint{46567}  \\
			ISCX-IDS2012         & \numprint{30917}  & \numprint{1097070}  & \numprint{165043} \\
			CIC-IDS2018      & \numprint{33176}  & \numprint{7948748}  & \numprint{38478} \\
			CIC-DDoS2019         & \numprint{1290}   & \numprint{20364525} & \numprint{12224} \\
			\bottomrule
		\end{tabular}
		\label{tab:Experiment.Dataset}
	\end{table}

\begin{table*}[!htb]
\centering
\caption{AUC and Running Time when detecting edge anomalies. Averaged over $5$ runs.}
\label{tab:edge}
\begin{tabular}{@{}lrrrrrrr}
\toprule
  Dataset & \densestream & \sedanspot & MIDAS-R & PENminer & F-FADE & \textbf{\methodedge-G} & \textbf{\methodedge-L} \\ 
\midrule
\multirow{2}{*}{DARPA} & $0.532$ & $0.647 \pm 0.006$ & $0.953 \pm 0.002$ &  0.872 & $0.919 \pm 0.005$ & $\bf 0.970 \pm 0.001$ & $ 0.964 \pm 0.001$ \\ 
& 57.7s & 129.1s  & 1.4s & 5.21 hrs & 317.8s & 28.7s & 6.1s \\ \midrule
\multirow{2}{*}{ISCX-IDS2012} & $0.551$ & $0.581 \pm 0.001$ & $0.820 \pm 0.050$ & 0.530 & $0.533 \pm 0.020$ & $\bf 0.954 \pm 0.000$ & $\bf 0.957 \pm 0.003$ \\ 
 & 138.6s & 19.5s & 5.3s & 1.3 hrs & 137.4s & 7.8s & 0.7s \\  \midrule
\multirow{2}{*}{CIC-IDS2018} & $0.756$ & $0.325 \pm 0.037$ & $0.919 \pm 0.019$ &  0.821 & $0.607 \pm 0.001$ & $\bf 0.963 \pm 0.014$ & $0.927 \pm 0.035$ \\ 
 & 3.3 hours & 209.6s & 1.1s & 10 hrs & 279.7s & 58.4s & 10.2s \\  \midrule
\multirow{2}{*}{CIC-DDoS2019} & $0.263$ & $0.567 \pm 0.004$ & $0.983 \pm 0.003$ &  --- & $0.717 \pm 0.041$ & $\bf 0.997 \pm 0.001$ & $\bf 0.998 \pm 0.001$ \\
 & 265.6s & 697.6s & 2.2s & > 24 hrs & 18.7s & 123.3s & 17.8s \\
\bottomrule
\end{tabular}
\end{table*}


In this section, we evaluate the performance of our approaches as compared to all baselines discussed in Table \ref{tab:comparison} and aim to answer the following questions:
\begin{enumerate}[label=\textbf{Q\arabic*.}]
\item {\bf Edge Anomalies:} How accurately do \methodedge-G and \methodedge-L detect edge anomalies compared to baselines? Are they fast and scalable?
\item {\bf Graph Anomalies:} How accurately do \methodgraph\ and \methodgraph-K detect graph anomalies i.e. anomalous graph snapshots? Are they fast and scalable?
\end{enumerate}

Table~\ref{tab:Experiment.Dataset} shows the statistical summary of the four real-world datasets that we use: \emph{DARPA} \citep{lippmann1999results} and \emph{ISCX-IDS2012} \citep{shiravi2012toward} are popular datasets for graph anomaly detection used by baselines to evaluate their algorithms; \citep{ring2019survey} surveys more than $30$ datasets and recommends to use the newer \emph{CIC-IDS2018} and \emph{CIC-DDoS2019} datasets \citep{sharafaldin2018toward,sharafaldin2019developing} containing modern attack scenarios. $|E|$ corresponds to the total number of edge records, $|V|$ and $|T|$ are the number of unique nodes and unique timestamps, respectively.

Similar to baseline papers, we report the Area under the ROC curve (AUC) and the running time. AUC is calculated by plotting the true positive rate (TPR) against the false positive rate (FPR) at various classification thresholds and then calculating the area under the resulting receiver operating characteristic (ROC) curve. The appropriate classification threshold for an anomaly detection system will depend on the specific application and the cost of false positives and false negatives, however, since AUC is independent of the classification threshold, one can evaluate the overall performance of the system without having to choose a specific threshold. Unless explicitly specified, all experiments including those on the baselines are repeated $5$ times and the mean is reported.

Appendix \ref{sec:setup} describes the experimental setup. Hyperparameters for the baselines are provided in Appendix \ref{sec:baselines}. All edge (or graph)-based methods output an anomaly score per edge (or graph), a higher score implying more anomalousness.

    

\subsection{Edge Anomalies}

{\bf Accuracy:} Table \ref{tab:edge} shows the AUC of edge anomaly detection baselines, \methodedge-G, and \methodedge-L. We report a single value for \densestream\ and PENminer because these are non-randomized methods. PENminer is unable to finish on the large \emph{CIC-DDoS2019} within 24 hours. \sedanspot\ uses personalised PageRank to detect anomalies and is not always able to detect anomalous edges occurring in dense block patterns while PENminer is unable to detect structural anomalies. Among the baselines, MIDAS-R is most accurate, however, it performs worse when there is a large number of timestamps as in \emph{ISCX-IDS2012}. Note that \methodedge-G and \methodedge-L outperform all baselines on all datasets.

{\bf Running Time:} Table \ref{tab:edge} shows the running time (excluding I/O) and real-time performance of \methodedge-G and \methodedge-L. Since \methodedge-L maintains a local dense submatrix, it is faster than \methodedge-G. \densestream\ maintains dense blocks incrementally for every coming tuple and updates dense subtensors when it meets an updating condition, limiting the detection speed. \sedanspot\ requires several subprocesses (hashing, random-walking, reordering, sampling, etc), PENminer and F-FADE need to actively extract patterns for every graph update, resulting in a large computation time. When there is a large number of timestamps like in \emph{ISCX-IDS2012}, MIDAS-R performs slower than \methodedge-L which is fastest.

{\bf AUC vs Running Time:} Figure \ref{fig:edgea} plots accuracy (AUC) vs. running time (log scale, in seconds, excluding I/O) on \emph{ISCX-IDS2012} dataset. \methodedge-G and \methodedge-L achieve much higher accuracy compared to all baselines, while also running significantly faster.

\begin{figure}[!htb]
    \centering
    \includegraphics[width=0.4\textwidth]{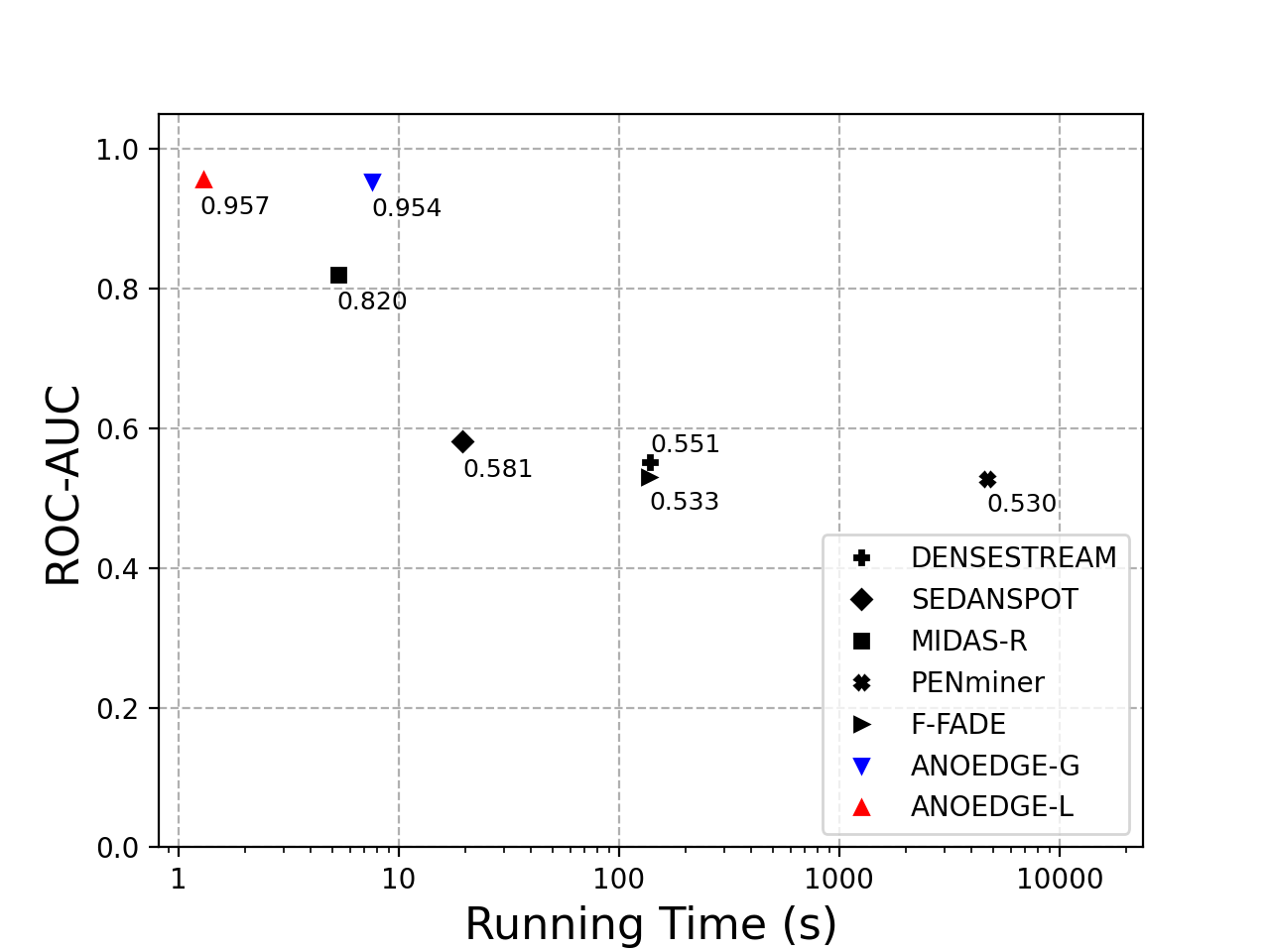}
  \caption{AUC vs running time when detecting edge anomalies on \emph{ISCX-IDS2012}}
  \label{fig:edgea}
\end{figure}

\begin{table*}[!htb]
\centering
\caption{AUC and Running Time when detecting graph anomalies. Averaged over $5$ runs.}
\label{tab:graph}
\begin{tabular}{@{}lrrrrr}
\toprule
Dataset & \densealert & \spotlight & \anomrank & \textbf{\methodgraph} & \textbf{\methodgraph-K} \\ 
\midrule
\multirow{2}{*}{DARPA} & $0.833$ & $0.728 \pm 0.016$  & $0.754$ & $0.835 \pm 0.002$  & $\bf 0.839 \pm 0.002$ \\ 
& 49.3s & 88.5s & 3.7s & 0.3s & 0.3s \\ \midrule
\multirow{2}{*}{ISCX-IDS2012} & $0.906$ & $0.872 \pm 0.019$ & $0.194$ & $\bf0.950 \pm 0.001$ & $\bf 0.950 \pm 0.001$ \\ 
& 6.4s & 21.1s & 5.2s & 0.5s & 0.5s \\ \midrule
\multirow{2}{*}{CIC-IDS2018}  & $0.950$  & $0.835 \pm 0.022$ & $0.783$ & $\bf 0.957 \pm 0.000$ & $\bf 0.957 \pm 0.000$ \\
& 67.9s & 149.0s & 7.0s & 0.2s & 0.3s \\ \midrule
\multirow{2}{*}{CIC-DDoS2019} & $0.764$ & $0.468 \pm 0.048$ & $0.241$ & $0.946 \pm 0.002$ & $\bf 0.948 \pm 0.002$ \\   
& 1065.0s & 289.7s & 0.2s & 0.4s & 0.4s \\ 
\bottomrule
\end{tabular}
\end{table*}

{\bf Scalability:} Figures \ref{fig:edge}(a) and \ref{fig:edge}(b) plot the running time with increasing number of hash functions and edges respectively, on the \emph{ISCX-IDS2012} dataset. This demonstrates the scalability of \methodedge-G and \methodedge-L.

{\bf \methodedge-G vs \methodedge-L:} \methodedge-G finds a \textbf{G}lobal dense submatrix and therefore is more accurate than \methodedge-L as shown in the performance on \emph{CIC-IDS2018}. \methodedge-L on the other hand maintains and updates a \textbf{L}ocal dense submatrix around the matrix element and therefore has better time complexity and scalability to larger datasets.

\begin{figure}[!htb]
  \centering
  \begin{subfigure}[b]{0.35\textwidth}
    \centering
    \includegraphics[width=\textwidth]{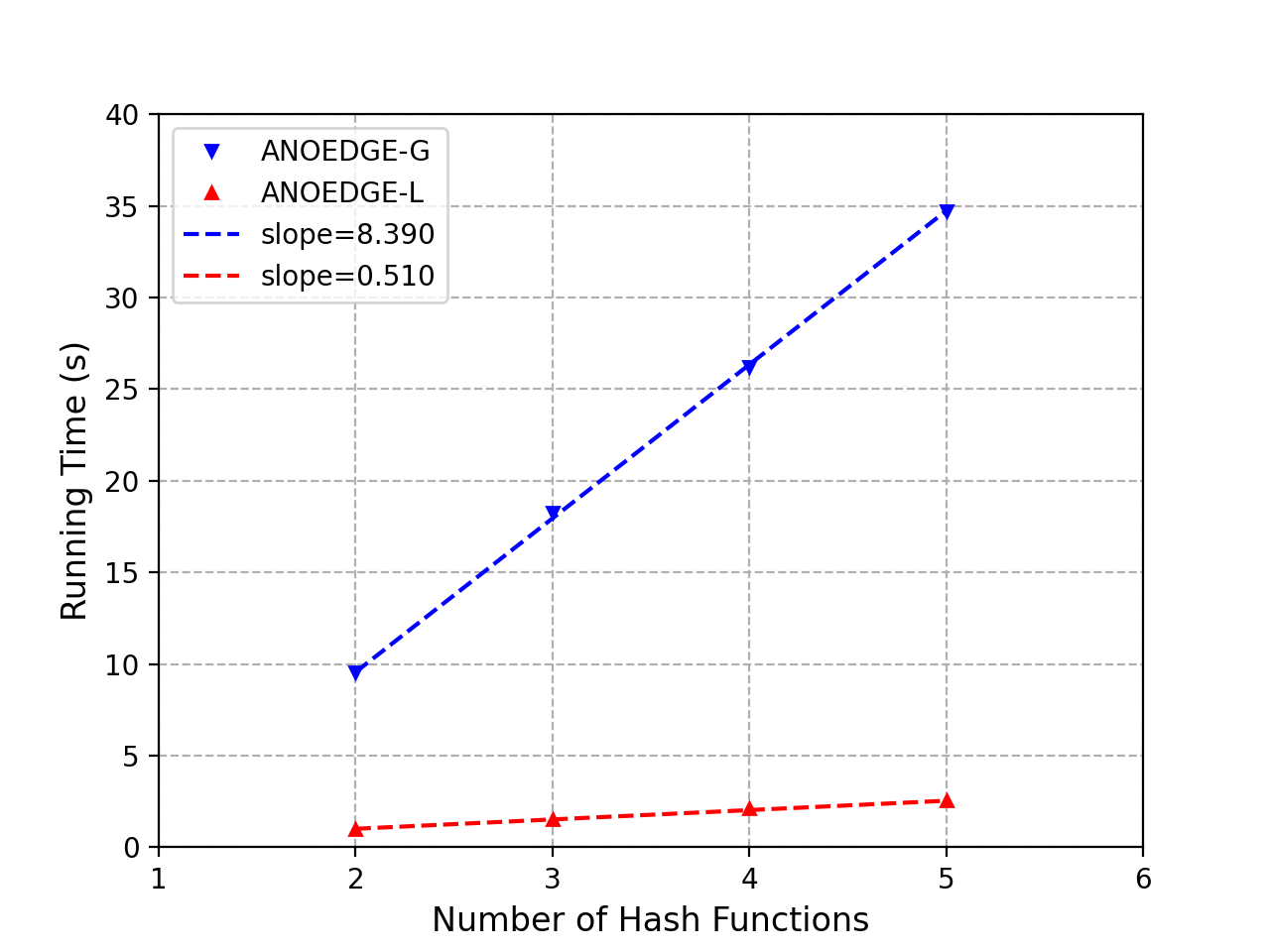}
  \end{subfigure}
  \begin{subfigure}[b]{0.35\textwidth}
    \centering
    \includegraphics[width=\textwidth]{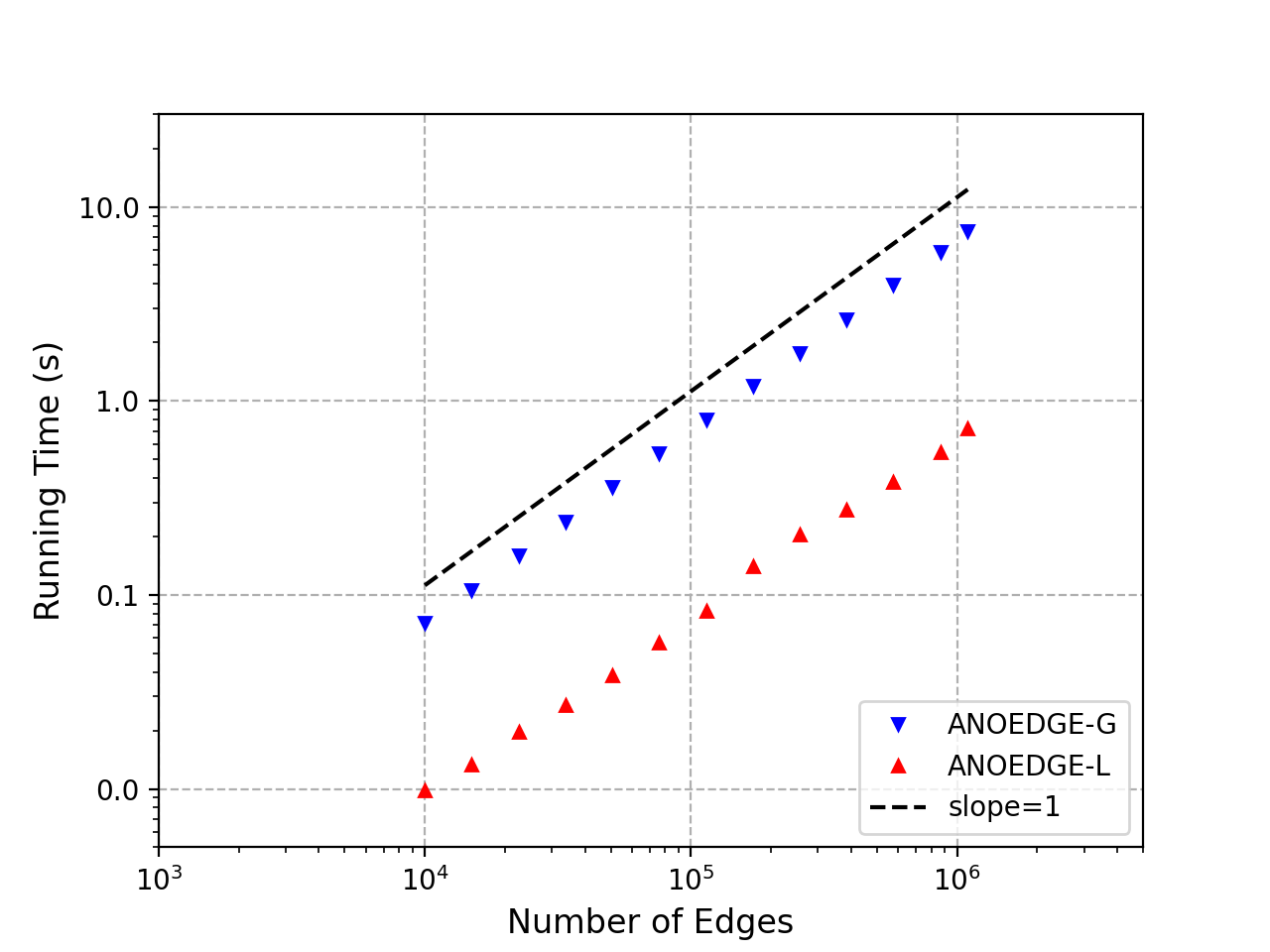}
  \end{subfigure}
  \caption{(a) Linear scalability with number of hash functions. (b) Linear scalability with number of edges.}
  \label{fig:edge}
\end{figure}

\subsection{Graph Anomalies}

{\bf Accuracy:} Table \ref{tab:graph} shows the AUC of graph anomaly detection baselines, \methodgraph, and \methodgraph-K. We report a single value for \densealert\ and \anomrank\ because these are non-randomized methods. \anomrank\ is not meant for a streaming scenario, therefore the low AUC. \densealert\ can estimate only one subtensor at a time and \spotlight\ uses a randomized approach without any actual search for dense subgraphs. Note that \methodgraph\ and \methodgraph-K outperform all baselines on all datasets while using a simple sketch data structure to incorporate dense subgraph search as opposed to the baselines.

{\bf Running Time:} Table \ref{tab:graph} shows the running time (excluding I/O). \densealert\ has $O(|\mathscr{E}|)$ worse case time complexity (per incoming edge). \anomrank\ needs to compute a global PageRank, which does not scale for stream processing. Note that \methodgraph\ and \methodgraph-K run much faster than all baselines.

{\bf AUC vs Running Time:} Figure \ref{fig:grapha} plots accuracy (AUC) vs. running time (log scale, in seconds, excluding I/O) on the \emph{CIC-DDoS2019} dataset. \methodgraph\ and \methodgraph-K achieve much higher accuracy compared to the baselines, while also running significantly faster.

\begin{figure}[!htb]
    \centering
    \includegraphics[width=0.4\textwidth]{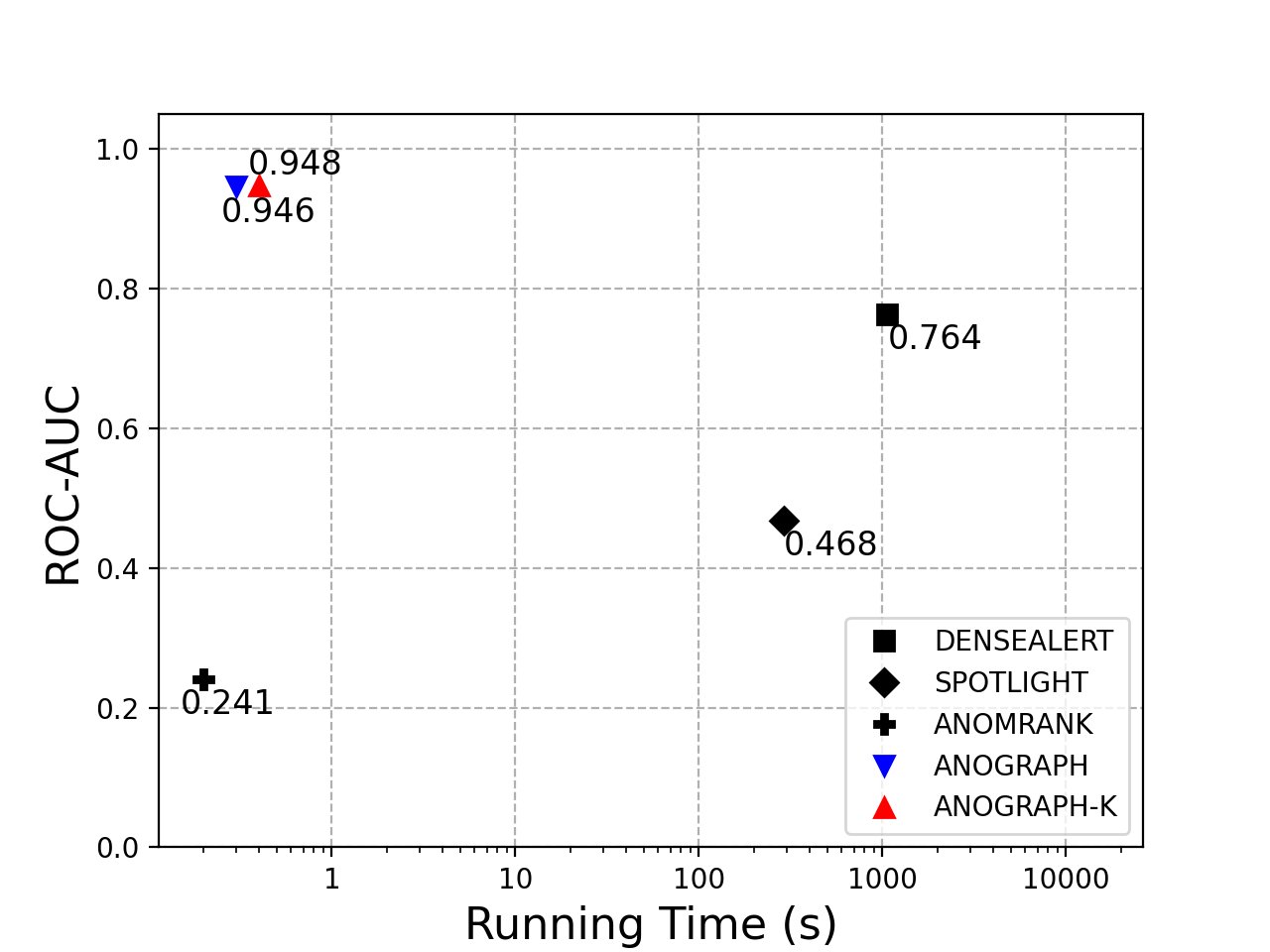}
  \caption{AUC vs running time when detecting graph anomalies on \emph{CIC-DDoS2019}}
  \label{fig:grapha}
\end{figure}

{\bf Scalability:} Figures \ref{fig:graph}(a), \ref{fig:graph}(b), and \ref{fig:graph}(c) plot the running time with increasing factor $K$ (used for top-$K$ in Algorithm \ref{alg:AnoGraph-K}), number of hash functions and number of edges respectively, on the \emph{CIC-DDoS2019} dataset. This demonstrates the scalability of \methodgraph\ and \methodgraph-K.

\begin{figure}[!htb]
  \centering
  \begin{subfigure}[b]{0.33\textwidth}
    \centering
    \includegraphics[width=\textwidth]{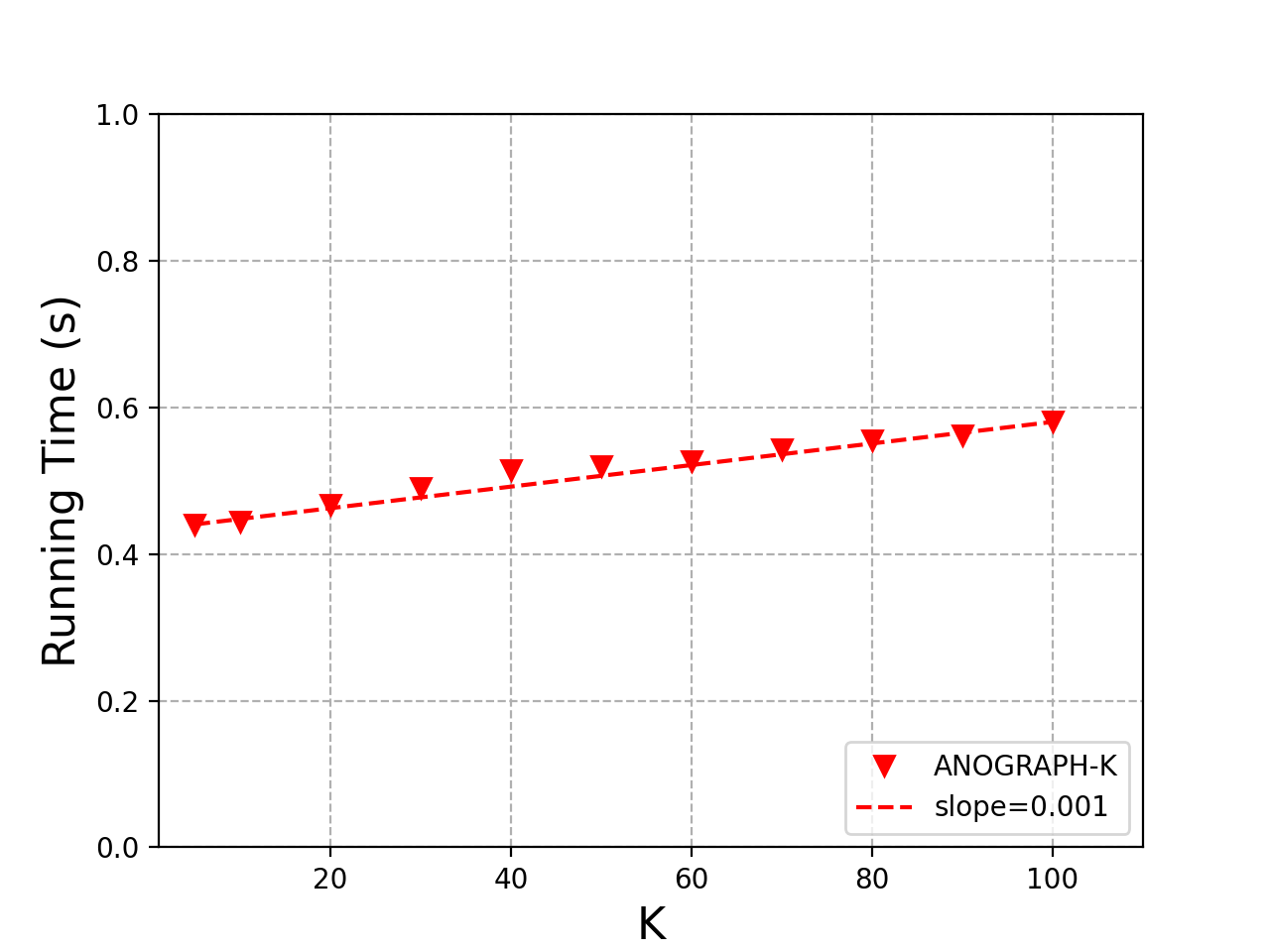}
  \end{subfigure}
  \begin{subfigure}[b]{0.33\textwidth}
    \centering
    \includegraphics[width=\textwidth]{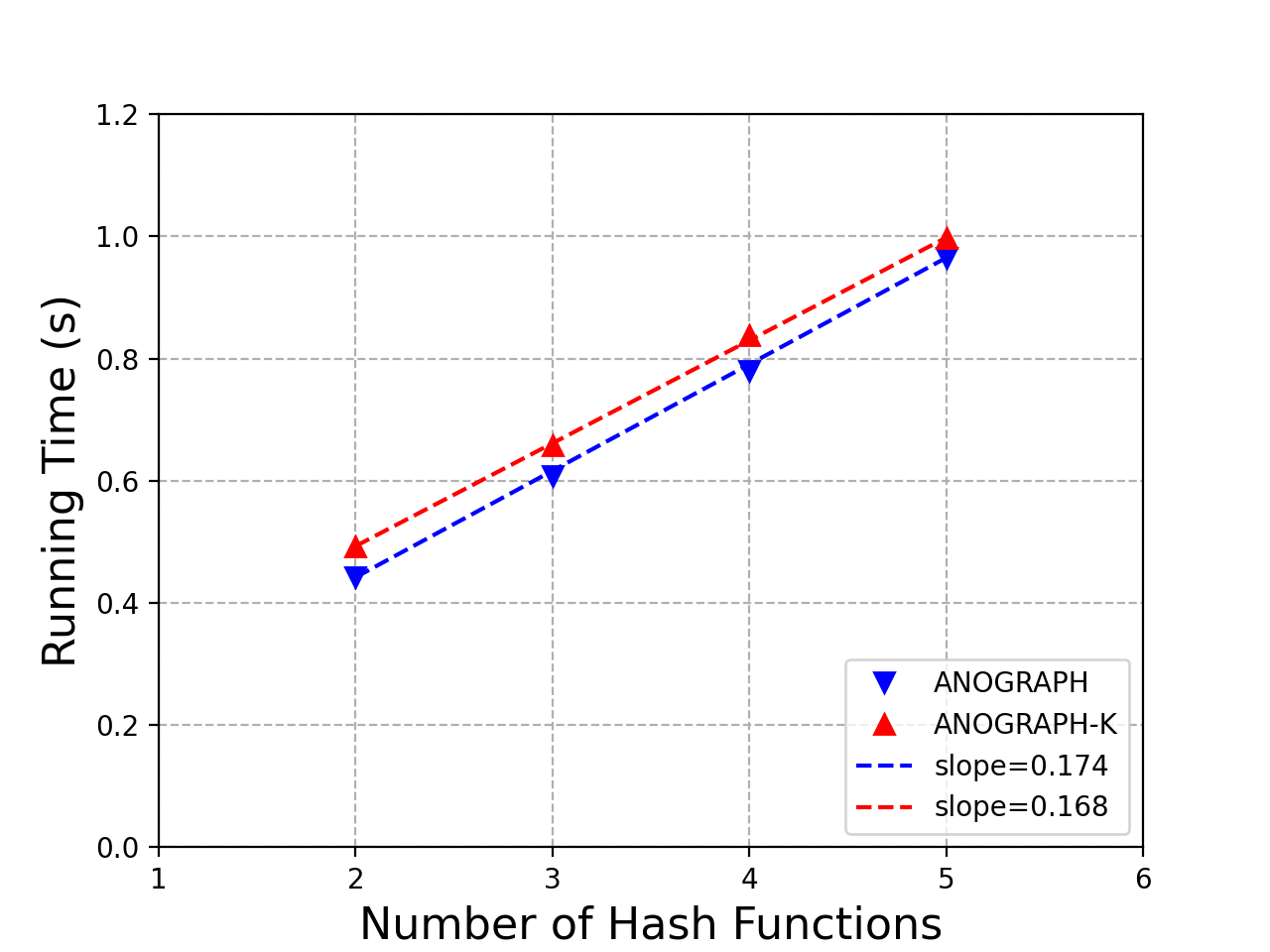}
  \end{subfigure}
    \begin{subfigure}[b]{0.33\textwidth}
    \centering
    \includegraphics[width=\textwidth]{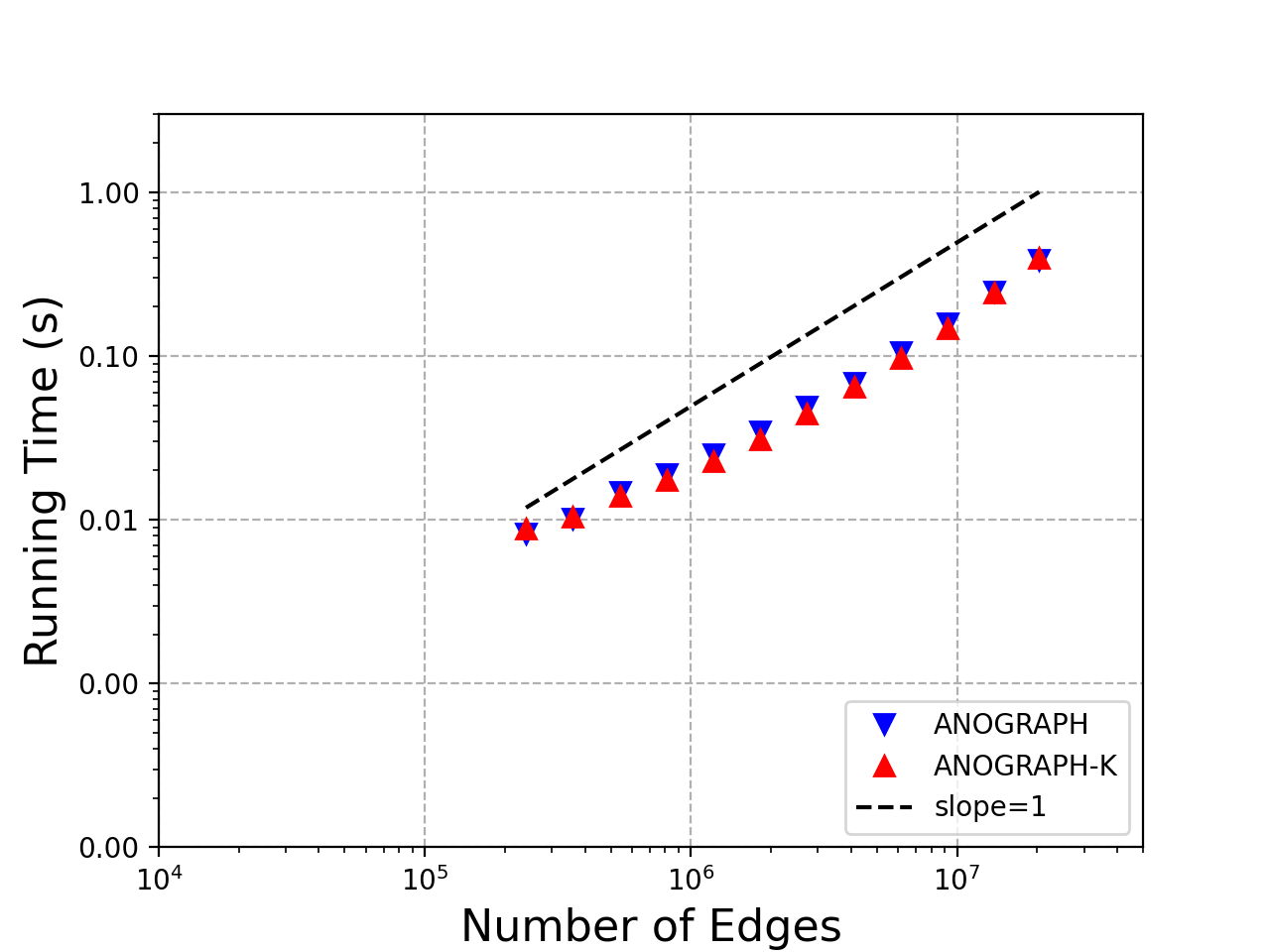}
  \end{subfigure}
  \caption{(a) \methodgraph-K scales linearly with factor $K$. (b) Linear scalability with number of hash functions. (c) Linear scalability with number of edges.}
  \label{fig:graph}
\end{figure}

{\bf \methodgraph\ vs \methodgraph-K:} \methodgraph\ greedily finds a dense submatrix and has a time complexity of $O(|\mathscr{G}|*n_r*n_b^2 + |\mathscr{E}|*n_r)$, while \methodgraph-K greedily finds a dense submatrix around \textbf{$K$} strategically picked matrix elements and has a time complexity of $O(|\mathscr{G}|*K*n_r*n_b^2 + |\mathscr{E}|*n_r)$. Therefore, \methodgraph-K is faster when $K$ is significantly smaller than $n$. \methodgraph-K is also more robust because it only considers a small number of matrix elements.

\begin{table}[!htb]
\centering
\caption{Influence of time window and edge threshold on the ROC-AUC when detecting graph anomalies.}
\label{tab:graphaucsupp}
\resizebox{\columnwidth}{!}{
\begin{tabular}{@{}lrrrr}
\toprule
Dataset & Time & Edge & \textbf{\methodgraph} & \textbf{\methodgraph-K} \\ 
 & Window & Threshold & & \\
\midrule
\multirow{4}{*}{DARPA} & $15$ & $25$ & $0.835 \pm 0.001$ & $0.838 \pm 0.001$ \\ 
 & $30$ & $50$ & $0.835 \pm 0.002$ & $0.839 \pm 0.002$ \\
 & $60$ & $50$ & $0.747 \pm 0.002$ & $0.748 \pm 0.001$ \\ 
 & $60$ & $100$ & $0.823 \pm 0.000$ & $0.825 \pm 0.001$ \\ 
\midrule
\multirow{4}{*}{ISCX-IDS2012} & $15$ & $25$ & $0.945 \pm 0.001$ & $0.945 \pm 0.000$ \\  
 & $30$ & $50$ & $0.949 \pm 0.001$ & $0.948 \pm 0.000$ \\ 
 & $60$ & $50$ & $0.935 \pm 0.002$ & $0.933 \pm 0.002$ \\ 
 & $60$ & $100$ & $0.950 \pm 0.001$ & $0.950 \pm 0.001$ \\ 
\midrule
\multirow{4}{*}{CIC-IDS2018} & $15$ & $25$ & $0.945 \pm 0.004$ & $0.947 \pm 0.006$ \\ 
& $30$ & $50$ & $0.959 \pm 0.000$ & $0.959 \pm 0.001$ \\
& $60$ & $50$  & $0.920 \pm 0.001$ & $0.920 \pm 0.001$ \\ 
& $60$ & $100$ & $0.957 \pm 0.000$ & $0.957 \pm 0.000$ \\ 
\midrule
\multirow{4}{*}{CIC-DDoS2019} & $15$ & $25$ & $0.864 \pm 0.002$ & $0.863 \pm 0.003$ \\  
& $30$ & $50$ & $0.861 \pm 0.003$ & $0.861 \pm 0.003$ \\
& $60$ & $50$  & $0.824 \pm 0.004$ & $0.825 \pm 0.005$ \\ 
& $60$ & $100$ & $0.946 \pm 0.002$ & $0.948 \pm 0.002$ \\ 
\bottomrule
\end{tabular}
}
\end{table}

\begin{table}[!htb]
		\centering
		\caption{Influence of temporal decay factor $\alpha$ on the ROC-AUC in \methodedge-G and \methodedge-L on \emph{DARPA}.}
		\label{tab:FactorVsAUC}
		\begin{tabular}{@{}lrr@{}}
			\toprule
			$\alpha$ & \methodedge-G & \methodedge-L \\
			\midrule
            $0.2$ & $0.964$ & $0.957$ \\
            $0.4$ & $0.966$ & $0.959$ \\
            $0.6$ & $0.968$ & $0.961$ \\
            $0.8$ & $0.969$ & $0.964$ \\
            $0.9$ & $0.969$ & $0.966$ \\
            $0.95$ & $0.966$ & $0.966$ \\
			\bottomrule
		\end{tabular}
\end{table}

\subsection{Hyperparameter Study}
\label{app:additionalresults}

Table \ref{tab:graphaucsupp} shows the performance of \methodgraph\ and \methodgraph-K for multiple time windows and edge thresholds. The edge threshold is varied in such a way that a sufficient number of anomalies are present within the time window. \methodgraph\ and \methodgraph-K achieve comparable results to those in Table \ref{tab:graph}. Table \ref{tab:FactorVsAUC} shows the robustness of \methodedge-G and \methodedge-L as we vary the temporal decay factor $\alpha$.

\section{Conclusion}
In this paper, we extend the CMS data structure to a higher-order sketch to capture complex relations in graph data and to reduce the problem of detecting suspicious dense subgraphs to finding a dense submatrix in constant time. We then propose four sketch-based streaming methods to detect edge and subgraph anomalies in constant update time and memory. Furthermore, our approach is the first streaming work that incorporates dense subgraph search to detect graph anomalies in constant memory and time. We also provide a theoretical guarantee on the submatrix density measure and prove the time and space complexities of all methods. Experimental results on four real-world datasets demonstrate our effectiveness as opposed to popular state-of-the-art streaming edge and graph baselines. Future work could consider incorporating rectangular H-CMS matrices, node and edge representations, more general types of data, including tensors, and parallel computing to process large dynamic graphs with a high volume of incoming edges.

\begin{acks}
This work was supported by the National Research Foundation Singapore, NCS Pte. Ltd., National University of Singapore under the NUS-NCS Joint Laboratory (Grant A-0008542-00-00), and by NSF under Grant SaTC-1930941.

\end{acks}

\clearpage

\bibliographystyle{ieeetr}
\balance
\bibliography{references}

\clearpage
\appendix
\section*{Appendix}


\setcounter{theorem}{0}
\setcounter{proposition}{0}


\section{H-CMS Estimate Guarantee Proof}
\label{app:proofshcms}
\begin{theorem} \label{thm:1supp}
For all $k \in [n_r]$, let $h_k(u,v) = (h'_k(u), h''_k(v))$ where each of hash functions $h'_k$ and $h''_k$ is chosen uniformly at random from a pairwise-independent family. Here, we allow both cases of $h'=h''$ and $h'\neq h''$. Fix $\delta>0$ and set  $n_r = \ceil*{\ln \frac{1}{\delta}}$ and  $n_b^{}=\ceil{\frac{e}{\epsilon}}$.  Then, with probability at least $1-\delta$, $ \hat y(u,v)\le y(u,v)+\epsilon M$.
\end{theorem}

\begin{proof}
Fix $j \in [n_r]$. Let  $a=(u_{a}, v_{a})$ and $b=(u_{b}, v_{b})$ such that $a\neq b$. This implies that at least one of the following holds: $u_a \neq u_b$ or  $v_a \neq v_b$. Since $h'_j$ (and $h''_j$) is chosen uniformly at random from a pairwise-independent family, 
$P(h'_{j}(u_{a}) = h'_{j}(u_{b})) = \frac{1}{n_b}$ or $P(h''_{j}(v_{a}) = h''_{j}(v_{b})) = \frac{1}{n_b}$. If $P(h'_{j}(u_{a}) = h'_{j}(u_{b})) = \frac{1}{n_b}$, we have that $P(h_j(a) = h_j(b))=P(h'_j(u_{a}) = h'_j(u_{b}) \wedge h''_j(v_{a}) = h''_j(v_{b}))= P(h'_j(u_{a}) = h'_j(u_{b}))P(h''_j(u_{a}) = h''_j(u_{b})|h'_j(u_{a}) = h'_j(u_{b}) ) \le \frac{1}{n_b}=\frac{\epsilon}{e}$. Similarly, if $P(h''_{j}(v_{a}) = h''_{j}(v_{b})) = \frac{1}{n_b}$,  $P(h_j(a) = h_j(b))= P(h'_j(u_{a}) = h'_j(u_{b})|h''_j(u_{a}) = h''_j(u_{b}))P(h''_j(u_{a}) = h''_j(u_{b}) ) \le \frac{1}{n_b}=\frac{\epsilon}{e}$. 
Thus, in the both cases, the probability of the collision is $P(h_{j}(a) = h_{j}(b))=\frac{\epsilon}{e}$. Thus, by defining $X_{a,j}=\sum_{b} \one\{a \neq b  \wedge h_{j}(a) = h_{j}(b)\}y(b)$,
$
\EE[X_{a,j}]\le\sum_{b}^{} y(b)\EE[\one\{a \neq b  \wedge h_{j}(a) = h_{j}(b)\}]\le \frac{\epsilon}{e}M. 
$     
Since $\hat y(a)= \min_j y(a)+X_{a,j}$, this implies that $P(\hat y(a)> y(a)+\epsilon M)) =P(\min_j y(a)+X_{a,j}> y(a)+\epsilon M))=P(\min_j X_{a,j}>\epsilon M))\le P(\min_j X_{a,j}> e\EE[X_{a,j}])) $.
By  the Markov's inequality on the right-hand side, we have that $P(\hat y(a)> y(a)+\epsilon M)) \le P(\min_j X_{a,j}> e\EE[X_{a,j}]))\le e^{-d}\le \delta$.
\end{proof}

\section{Edge Anomalies-Proofs}

\subsection{\methodedge-G}
\label{app:proofsAnoEdge-G}
\begin{proposition}\label{thm:AnoEdge-G-time-supp}
Time complexity of Algorithm \ref{alg:AnoEdge-G} is $O(|\mathscr{E}|*n_r*n_b^2)$. Memory complexity of Algorithm \ref{alg:AnoEdge-G} is $O(n_r*n_b^2)$.
\end{proposition}
\begin{proof}
Procedure \textsc{Edge-Submatrix-Density} removes rows (or columns) iteratively, and the total number of rows and columns that can be removed is $n_b + n_b - 2$. In each iteration, the approach performs the following three operations: (a) pick the row with minimum row-sum; (b) pick the column with minimum column-sum; (c) calculate density. We keep $n_b$-sized arrays for flagging removed rows (or columns), and for maintaining row-sums (or column-sums). Operations (a) and (b) take maximum $n_b$ steps to pick and flag the row with minimum row-sum (or column-sum). Updating the column-sums (or rows-sums) based on the picked row (or column) again takes maximum $n_b$ steps. Time complexity of (a) and (b) is therefore $O(n_b)$. Density is directly calculated based on subtracting the removed row-sum (or column-sum) and reducing the row-count (or column-count) from the earlier density value. Row-count and column-count are kept as separate variables. Therefore, the time complexity of the density calculation step is $O(1)$. Total time complexity of procedure \textsc{Edge-Submatrix-Density} is $O((n_b + n_b - 2)*(n_b + n_b + 1)) = O(n_b^2)$.

Time complexity to initialize and decay the H-CMS data structure is $O(n_r*n_b^2)$. Temporal decay operation is applied whenever the timestamp changes, and not for every received edge. Update counts operation updates a matrix element value ($O(1)$ operation) for $n_r$ matrices, and the time complexity of this step is $O(n_r)$. Anomaly score for each edge is based on the submatrix density computation procedure which is $O(n_b^2)$; the time complexity of $n_r$ matrices becomes $O(n_r*n_b^2)$. Therefore, the total time complexity of Algorithm \ref{alg:AnoEdge-G} is $O(|\mathscr{E}|*(n_r + n_r*n_b^2)) = O(|\mathscr{E}|*n_r*n_b^2)$.

For procedure \textsc{Edge-Submatrix-Density}, we keep an $n_b$-sized arrays to flag rows and columns that are part of the current submatrix, and to maintain row-sums and column-sums. Total memory complexity of \textsc{Edge-Submatrix-Density} procedure is $O(4*n_b) = O(n_b)$.

Memory complexity of H-CMS data structure is $O(n_r*n_b^2)$. Dense submatrix search and density computation procedure require $O(n_b)$ memory. For $n_r$ matrices, this becomes $O(n_r*n_b)$. Therefore, the total memory complexity of Algorithm \ref{alg:AnoEdge-G} is $O(n_r*n_b^2 + n_r*n_b) = O(n_r*n_b^2)$.
\end{proof}

\subsection{\methodedge-L}
\label{app:proofsAnoEdge-L}

\begin{proposition}\label{thm:AnoEdge-L-time-supp}
Time complexity of Algorithm \ref{alg:AnoEdge-L} is $O(n_r*n_b^2 + |\mathscr{E}|*n_r*n_b)$. Memory complexity of Algorithm \ref{alg:AnoEdge-L} is $O(n_r*n_b^2)$.
\end{proposition}
\begin{proof}
As shown in Proposition \ref{thm:AnoEdge-G-time}, the time complexity of H-CMS is $O(n_r*n_b^2)$ and update operation is $O(n_r)$. Current submatrix $(S_{cur}, T_{cur})$ is updated based on \emph{expand} and \emph{condense} submatrix operations. (a) We keep an $n_b$-sized array to flag the current submatrix rows (or column), and also to maintain row-sums (or column-sums). Expand submatrix operation depends on the elements from row $h(u)$ and column $h(v)$, and the density is calculated by considering these elements, thus requiring maximum $n_b$ steps. Upon addition of the row (or column), the dependent column-sums (or row-sums) are also updated taking maximum $n_b$ steps. Time complexity of expand operation is therefore $O(n_b)$. (b) Condense submatrix operation removes rows and columns iteratively. A row (or column) elimination is performed by selecting the row (or column) with minimum row-sum (or column-sum) in $O(n_b)$ time. Removed row (or column) affects the dependent column-sums (or row-sums) and are updated in $O(n_b)$ time. Time complexity of a row (or column) removal is therefore $O(n_b)$. Condense submatrix removes rows (or columns) that were once added by the expand submatrix operation which in worse case is $O|\mathscr{E}|$.

Expand and condense submatrix operations are performed for $n_r$ matrices. Likelihood score calculation depends on elements from row $h(u)$ and column $h(v)$, and takes $O(n_r*n_b)$ time for $n_r$ matrices. Therefore, the total time complexity of Algorithm \ref{alg:AnoEdge-L} is $O(n_r*n_b^2 + |\mathscr{E}|*n_r + |\mathscr{E}|*n_r*n_b + |\mathscr{E}|*n_r*n_b + |\mathscr{E}|*n_r*n_b) = O(n_r*n_b^2 + |\mathscr{E}|*n_r*n_b)$.

Memory complexity of the H-CMS data structure is $O(n_r*n_b^2)$. To keep current submatrix information, we utilize $n_b$-sized arrays similar to Proposition \ref{thm:AnoEdge-G-time}. For $n_r$ matrices, submatrix information requires $O(n_r*n_b)$ memory. Hence, total memory complexity of Algorithm \ref{alg:AnoEdge-L} is $O(n_r*n_b^2 + n_r*n_b) = O(n_r*n_b^2)$.
\end{proof}

\section{Graph Anomalies-Proofs}

\subsection{\methodgraph}
\label{app:proofsAnoGraph}

\begin{proposition}\label{thm:AnoGraph-time-supp}
Time complexity of Algorithm \ref{alg:AnoGraph} is $O(|\mathscr{G}|*n_r*n_b^2 + |\mathscr{E}|*n_r)$. Memory complexity of Algorithm \ref{alg:AnoGraph} is $O(n_r*n_b^2)$.
\end{proposition}
\begin{proof}
Procedure \textsc{\methodgraph-Density} iteratively removes row (or column) with minimum row-sum (or column-sum). Maximum number of rows and columns that can be removed is $n_b + n_b - 2$. We keep $n_b$-sized arrays to store the current submatrix rows and columns, and row-sums and column-sums. At each iteration, selecting the row (or column) with minimum row-sum (or column-sum) takes $O(n_b)$ time, and updating the dependent row-sums (or column-sums) also $O(n_b)$ time. Density is calculated in $O(n_b)$ time based on the current submatrix row-sum and column-sum. Each iteration takes $O(n_b + n_b + n_b) = O(n_b)$ time. Hence, the total time complexity of \textsc{\methodgraph-Density} procedure is $O((n_b + n_b - 2)*n_b) = O(n_b^2)$.

Initializing the H-CMS data structure takes $O(n_r*n_b^2)$ time. When a graph arrives, \methodgraph: (a) resets counts that take $O(n_r*n_b^2)$ time; (b) updates counts taking $O(1)$ time for every edge update; (c) computes submatrix density that follows from procedure \textsc{\methodgraph-Density} and takes $O(n_b^2)$ time. Each of these operations is applied for $n_r$ matrices. Therefore, the total time complexity of Algorithm \ref{alg:AnoGraph} is $O(n_r*n_b^2 + |\mathscr{G}|*n_r*n_b^2 + |\mathscr{E}|*n_r + |\mathscr{G}|*n_r*n_b^2) = O(|\mathscr{G}|*n_r*n_b^2 + |\mathscr{E}|*n_r)$, where $|\mathscr{E}|$ is the total number of edges over graphs $\mathscr{G}$.

For procedure \textsc{\methodgraph-Density}, we keep $n_b$-sized array to flag rows and columns that are part of the current submatrix, and to maintain row-sums and column-sums. Hence, memory complexity of \textsc{\methodgraph-Density} procedure is $O(4*n_b) = O(n_b)$. 

H-CMS data structure requires $O(n_r*n_b^2)$ memory. Density computation relies on \textsc{\methodgraph-Density} procedure, and takes $O(n_b)$ memory. Therefore, the total memory complexity of Algorithm \ref{alg:AnoGraph} is $O(n_r*n_b^2)$.
\end{proof}


\subsection{Proofs: \methodgraph-K}
\label{app:proofsAnoGraph-K}

\begin{proposition}\label{thm:AnoGraph-K-time-supp}
Time complexity of Algorithm \ref{alg:AnoGraph-K} is $O(|\mathscr{G}|*K*n_r*n_b^2 + |\mathscr{E}|*n_r)$. Memory complexity of Algorithm \ref{alg:AnoGraph-K} is $O(n_r*n_b^2)$.
\end{proposition}
\begin{proof}
Relevant operations in Procedure \textsc{\methodgraph-K-Density} directly follow from \textsc{Edge-Submatrix-Density} procedure, which has $O(n_b^2)$ time complexity. \textsc{Edge-Submatrix-Density} procedure is called $K$ times, therefore, the total time complexity of \textsc{\methodgraph-K-Density} procedure is $O(K*n_b^2)$.

For Algorithm \ref{alg:AnoGraph-K}, we initialize an H-CMS data structure that takes $O(n_r*n_b^2)$ time. When a graph arrives, \methodgraph-K: (a) resets counts that take $O(n_r*n_b^2)$ time; (b) updates counts taking $O(1)$ time for every edge update; (c) computes submatrix density that follows from procedure \textsc{\methodgraph-K-Density} and takes $O(K*n_b^2)$ time. Each of these operations is applied for $n_r$ matrices. Therefore, the total time complexity of Algorithm \ref{alg:AnoGraph-K} is $O(n_r*n_b^2 + |\mathscr{G}|*K*n_r*n_b^2 + |\mathscr{E}|*n_r + |\mathscr{G}|*n_r*n_b^2) = O(|\mathscr{G}|*K*n_r*n_b^2 + |\mathscr{E}|*n_r)$, where $|\mathscr{E}|$ is the total number of edges over graphs $\mathscr{G}$.

The density of $K$ submatrices is computed independently, and the memory complexity of Algorithm procedure \textsc{\methodgraph-K-Density} is the same as the memory complexity of \textsc{Edge-Submatrix-Density} procedure i.e. $O(n_b)$.

Maintaining the H-CMS data structure requires $O(n_r*n_b^2)$ memory. Density computation relies on \textsc{\methodgraph-K-Density} procedure, and it requires $O(n_b)$ memory. Therefore, the total memory complexity of Algorithm \ref{alg:AnoGraph-K} is $O(n_r*n_b^2)$.
\end{proof}

\section{Experimental Setup}\label{sec:setup}
	All experiments are carried out on a $2.4 GHz$ Intel Core $i9$ processor, $32 GB$ RAM, running OS $X$ $10.15.3$. For our approach, we keep $n_r=2$ and $n_b=32$ to have a fair comparison to MIDAS which uses ${n_b}^2=1024$ buckets. Temporal decay factor $\alpha=0.9$ for Algorithms $\ref{alg:AnoEdge-G}$ and $\ref{alg:AnoEdge-L}$. We keep $K=5$ for Algorithm $\ref{alg:AnoGraph-K}$. AUC for graph anomalies is shown with edge thresholds as $50$ for \emph{DARPA} and $100$ for other datasets. Time window is taken as $30$ minutes for \emph{DARPA} and $60$ minutes for other datasets.

\section{Baselines}\label{sec:baselines}

We use open-source implementations of \densestream\ \citep{shin2017densealert}, \sedanspot\ \citep{eswaran2018sedanspot}, MIDAS-R \citep{bhatia2020midas} (C++), PENminer \citep{belth2020mining}, F-FADE \citep{chang2021f}, \densealert\ \citep{shin2017densealert}, and \anomrank\ \citep{yoon2019fast} provided by the authors, following parameter settings as suggested in the original paper. For \spotlight\ \citep{eswaran2018spotlight}, we used open-sourced implementations of Random Cut Forest \citep{awsrando88:online} and Carter Wegman hashing \citep{carter1979universal}.

    
    \subsection{Edge Anomalies}
    
    \begin{enumerate}
    \item {{\sedanspot:}}
		\texttt{sample\_size} $=10000$,
		\texttt{num\_walk} $=50$,
		\texttt{restart\_prob} $0.15$

    \item{{MIDAS:}} The size of CMSs is 2 rows by 1024 columns for all the tests.
	For MIDAS-R, the decay factor $\alpha=0.6$.

	\item{{PENminer:}}
		\texttt{ws} $=1$,
		\texttt{ms} $=1$,
		\texttt{view} = \texttt{id},
		\texttt{alpha} $=1$,
		\texttt{beta} $=1$,
		\texttt{gamma} $=1$

	\item{{\densestream:}} We keep default parameters, i.e., order $=3$.
	
	\item{{F-FADE:}}
		\texttt{embedding\_size} $=200$,
		\texttt{W\_upd} $=720$,
		\texttt{T\_th} $=120$,
		\texttt{alpha} $=0.999$,
		\texttt{M} $=100$

	For \texttt{t\_setup}, we always use the timestamp value at the $10^{th}$ percentile of the dataset.
	\end{enumerate}
	
	\subsection{Graph Anomalies}
	
	\begin{enumerate}

	\item{{\spotlight:}}
	\texttt{K} $=50$,
	\texttt{p} $=0.2$,
	\texttt{q} $=0.2$

	\item{{\densealert:}} We keep default parameters, i.e., order $=3$ and window=$60$.
	
	\item{{\anomrank:}} We keep default parameters, i.e., damping factor c $= 0.5$, and L1 changes of node score vectors threshold epsilon $= 10^{-3}$.
	We keep ${1/4}^{th}$ number of graphs for initializing mean/variance as mentioned in the respective paper.
	
    \end{enumerate}

\end{document}